\documentclass[letterpaper, 10 pt, conference]{ieeeconf}  % Comment this line out
\IEEEoverridecommandlockouts                              % This command is only
\usepackage{graphics} % for pdf, bitmapped graphics files
\usepackage{epsfig} % for postscript graphics files
\usepackage{mathptmx} % assumes new font selection scheme installed
\usepackage{times} % assumes new font selection scheme installed
\usepackage{amsmath} % assumes amsmath package installed
\usepackage{amssymb}% assumes amsmath package installed

\usepackage{amsthm}
\usepackage{mathtools}
\theoremstyle{remark}
\newtheorem{assumption}{Assumption}

\newtheorem{theorem}{Theorem}
\newtheorem{defi}{Definition}
\usepackage{cite}
\usepackage{algorithmic}
\usepackage{algorithmic}
\usepackage{textcomp}
\usepackage{xcolor}
\usepackage{comment}
\usepackage{url}
\usepackage{booktabs}
\usepackage{tikz}
\usepackage{hyperref}

\title{\LARGE \bf
Cyberattack Detection for Nonlinear Leader-Following Multi-Agent Systems Using Set-Membership Fuzzy Filtering
}

\author{Mahshid Rahimifard, Amir~M.~Moradi Sizkouhi, and Rastko~R.~Selmic,~\IEEEmembership{Senior Member,~IEEE}% <-this % stops a space
\thanks{M. Rahimifard, A. M. Moradi Sizkouhi, and R. R. Selmic are with with the Department of Electrical and Computer Engineering, Concordia University, Montreal, QC, Canada.
        {\tt\small seyedehmahshid.rahimifard@concordia.ca},
        {\tt\small amirmohammad.moradis@concordia.ca},
        {\tt\small rastko.selmic@concordia.ca}.%
}
}

\begin{document}

\maketitle
\thispagestyle{empty}
\pagestyle{empty}

%%%%%%%%%%%%%%%%%%%%%%%%%%%%%%%%%%%%%%%%%%%%%%%%%%%%%%%%%%%%%%%%%%%%%%%%%%%%%%%%
\begin{abstract}

This paper is concerned with cyberattack detection in discrete-time, leader-following, nonlinear, multi-agent systems subject to unknown but bounded (UBB) system noises. The Takagi–Sugeno (T-S) fuzzy model is employed to approximate the nonlinear systems over the true value of the state. A distributed cyberattack detection method, based on a new fuzzy set-membership filtering method, which consists of two steps, namely a prediction step and a measurement update step, is developed for each agent to identify two types of cyberattacks at the time of their occurrence. The attacks are replay attacks and false data injection attacks affecting the leader-following consensus. We calculate an estimation ellipsoid set by updating the prediction ellipsoid set with the current sensor measurement data. Two criteria are provided to detect cyberattacks based on the intersection between the ellipsoid sets. If there is no intersection between the prediction set and the estimation set of an agent at the current time instant, a cyberattack on its sensors is declared. Control signal or communication signal data of an agent are under a cyberattack if its prediction set has no intersection with the estimation set updated at the previous time instant. Recursive algorithms for solving the consensus protocol and calculating the two ellipsoid sets for detecting attacks are proposed. Simulation results are provided to demonstrate the effectiveness of the proposed method.

\end{abstract}

%%%%%%%%%%%%%%%%%%%%%%%%%%%%%%%%%%%%%%%%%%%%%%%%%%%%%%%%%%%%%%%%%%%%%%%%%%%%%%%%
\section{INTRODUCTION}

Multi-Agent Systems (MAS) have wide range of applications such as internet of things (IoT), electrical grids, water distribution systems, transportation systems, autonomous vehicles and Unmanned Aerial Vehicles (UAVs) \cite{b2}. Reaching consensus in a distributed manner is a fundamental problem in MAS. Some distributed and decentralized methods for attack detection have been proposed \cite{b13,b14,b15,b16}. The agents transmit their data to neighboring agents through communication channels in distributed consensus protocols, and these channels are vulnerable to cyberattacks.

In \cite{b13} a distributed method to detect attacks in the communication network for the distributed control of interconnected systems has been proposed. In this study, only the local knowledge of the system is needed. However, through this method, stealthy attacks cannot be detected.
In \cite{b18,b19,b20,b21,b22}, secure designs against Denial of Service (DoS) attacks in a centralized setting have been studied. In \cite{b3}, an investigation on a distributed event-triggered secure cooperative control of linear multi-agent systems under DoS attacks has been conducted. This paper studies how to achieve average consensus in the presence of DoS attacks and does not detect the attacks. 

%In \cite{b4}, decentralized and real-time fault detection frameworks have been proposed for two different system models which are communication-based model and relative sensing-based model. The system under this study is homogeneous. Also, the proposed method is not able to detect the type of faults. Moreover, faulty agents will not achieve the system global task.

The problem of distributed simultaneous fault detection and leader-following consensus control for multi-agent systems has been investigated in \cite{b7}. In this study, a single module is used that conducts both tasks of fault detection and control objectives, simultaneously. Also, the proposed fault isolation method detect the faulty agent as well as the kind of fault. However, this method is restricted to the actuator and sensor faults to the linear systems with undirected topology. The authors of \cite{b5} have studied the leader-following consensus problem for heterogeneous multi-agent systems subject to both sensor and actuator attacks. In this study, the system is linear and there have been considered some restricting conditions.

%In \cite{b6}, an LMI based approach for the mixed $H_{\infty}/H_{-}$ distributed simultaneous fault detection and consensus control design of continuous-time linear multi-agent systems using relative output information has been proposed. However, there are some constraints to this study. First, the network topology is undirected. Second, the multi-agent system is considered to be homogeneous and also, the system is considered to be linear. Third, multiple faults cannot occur simultaneously in the team.

In \cite{b9} a state-dependent event-triggered control strategy for time-varying MASs over a finite horizon has been designed for the first time. The only attack considered in this study is the false data injection attack to the linear multi-agent system with undirected topology.

%In \cite{b11}, the problem of simultaneous fault detection and tracking control for a group of linear agents with event-triggered communication has been studied. In this method, each agent is able to detect both its own fault and the faults in its neighbours at the same time. The observer in this work is complex for the purpose of better performance indices and transmission rates.

Most of attack detection approaches, which are based on the state estimation, necessitate systems noises to be in a stochastic framework, and this leads to a probabilistic state estimation. For many real-world applications, accuracy in the state estimation is crucial. However, estimation based on probabilistic approach, such as Kalman filtering method, necessitates the use of mean and variance to describe the state distributions modelled as random variables (usually white and Gaussian perturbations). Consequently, considering unknown but bounded (UBB) noises is a much more appropriate approach to modeling state distributions.

Additionally, a common attack detection method, called the performance index test ($\chi^2$-detector), uses a residual signal to determine if the estimated behavior differs from that predicted by a model. Due to the nature of the Kalman filtering technique, the estimated and predicted states are single vectors and as a result, they cannot guarantee that a state is included in some region. Also, as the resulting UBB noises are sub-optimal for Kalman-type filtering, the reliability of attack detection is decreased. As a result of the need for set-valued estimation, the ellipsoidal state estimation technique was developed \cite{b34}. This method, known as the set-membership or set-valued state estimation filtering approach, has been extensively studied in filtering problems \cite{b35,b36,b37,b38,b39,b40,b41} and provides a set of state estimates in state space that contains the system's true state \cite{b42,b43}. By using convex optimization approaches, an optimal ellipsoid with minimal size can be determined for set-membership estimation, improving state estimation and detection performance.

The authors of \cite{b8} have studied a cyberattack detection method for the linear networked control systems through which for the first time, simultaneously, using the set-membership filtering for the purpose of the attack detection and distinguishing attacks on control signals from attacks on measurement outputs have been considered. However, they only considered the attack detection problem and there is no approach to the control of the system and the system is a single agent.

Except a few publications \cite{b47,b27}, most research on set-membership filtering considers linear systems \cite{b44,b45,b33,b8}. Linearization should best fit the nonlinear functions over a state estimate set rather than a state estimate point when we use the set-membership framework. The authors of \cite{b47}, linearized the nonlinear dynamics around the current estimate, then bounded the remaining terms by using interval mathematics and finally incorporated the remaining bounds as additions to the process or measurement noise bounds. Due to linearization around the estimated value of the state rather than the true value, the above approximations, bring a base point error \cite{b49}.

There are few works on the detection of replay attacks and to the best of our knowledge, all the existing works have been only done on linear systems. Therefore, detection of these attacks for the nonlinear systems are of prime importance since real-world systems are mostly nonlinear. The fuzzy model of Takagi-Sugeno (T-S) is an effective and universal approximator for a certain class of nonlinear dynamic systems.

Therefore, in this paper, we use it to approximate nonlinear systems \cite{b50,b51}. We linearize the nonlinear systems over the true value of state and eliminate the base point error. Our objective is to design a simultaneous distributed attack detection strategy and leader-following consensus control based on a new two-step fuzzy set-membership filtering approach in a distributed framework. By utilizing the fuzzy modeling approach and the S-procedure technique \cite{b28}, we determine bounding ellipsoidal sets for each agent by a recursive algorithm in state-space which guarantee the always enclosing of the system's true state \cite{b42,b43}, regardless of UBB noises, assuming no attacks are being made on the agent.

Each agent has a prediction and a measurement update step in its state estimation algorithm. The following two criteria are then used to detect cyberattacks:
\begin{enumerate}
    \item When a cyberattack violates the control signal of any agent, the prediction ellipsoid set of that agent and its estimation ellipsoid set, updated with the previous measurement output, do not intersect.
    \item When a cyberattack violates the sensor signal of any agent, the prediction ellipsoid set of that agent and its estimation ellipsoid set, updated at the current time instant, do not intersect.
\end{enumerate}

%In \cite{b17}, a distributed fault detection methodology has been presented for large-scale interconnected systems. This method incorporates a synchronization methodology with a filtering approach in order to reduce the effect of measurement noise and time delays on the fault detection performance. Also, a consensus-based estimator with time-varying weights has been introduced to improve fault detectability when there are variables shared among more than one subsystem. By using this methodology, we can know which process or actuator has been affected by a fault.

%There are some assumptions on the existing works that may restrict their applicability. These assumptions have some consequences. For example, some faults may become undetectable, detection times may become longer because of the delays in information exchange, and lack of accurate and timely information may lead to false-alarms. Although the methodology presented in \cite{b12} addresses these limitations, it has some other weaknesses. For instance, it considers the case when only one fault is affecting the system. Also, only the process and actuator faults are considered.

Comparing with the previous works, the contributions of our work are as follows:
\begin{itemize}
\item To the best of our knowledge, we studied the attack detection problem of nonlinear multi-agent systems subject to replay attacks for the first time in the literature.
\item We developed the fuzzy set-membership filtering approach for detection of the attacks.
\item We considered false data injection attacks on the control signal and communication networks, as well as replay attacks on the sensor measurement data.
\item We can distinguish attacks on control signals from attacks on measurement outputs.
\item We also are able to mitigate the effects of the attacks and recover the system performance.
\item Moreover, we ensure that we achieve the control goal which is achieving all the agent states to the leader-following consensus.
\end{itemize}
\section{Problem Formulation}
Interaction and communication is modelled as a connected directed graph $\mathcal{G}=\{\mathcal{V} ,\mathcal{E} , \mathcal{A}\}, \mathcal{V}=\{1,2,...,N\}, \mathcal{E}=\{i,j, \left(i,j\right) \in \mathcal{V}\}$  and $\mathcal{A}=\left(a_{ij}\right)\in \mathbb{R}^{N\times N}$which are the vertex set, the directed edge set and the weighted adjacency matrix of $\mathcal{G}$, respectively. The weights are defined as  $a_{ij}>0$, if $(j,i)\in\mathcal{E}$ and $a_{ij} = 0$, otherwise. A node from which an edge goes to node $i$ is a neighbor of node $i$. The set of the neighbors of node $i$ are indicated by $N_i$, where $N_i=\{j|\left(j,i\right)\in\mathcal{E}$. Moreover, the Laplacian matrix $\mathcal{L}=\left(l_{i,j}\right)\in\mathbb{R}^{N\times N}$ is defined as $\mathcal{L}=\mathcal{D}-\mathcal{A}$ and $\mathcal{D}=\operatorname{diag}_{N}^{i}\left\{d_{i}\right\}$, with $d_{i}=\sum_{j=1}^{N} a_{i j}$.

Consider a discrete-time nonlinear multi-agent system with $N$ agents, and the dynamics of agent, $i,i\in\{1,...,N\}$ is given as
  \begin{equation}
   \begin{cases}
          {x}_{i}(k+1)=f_i(x_i(k))+G_iu_i(k)+I_{i}(x_{i})\omega_i(k)\\
           y_i(k)=h_{i}(x_i(k))+F_{i}(x_{i})v_{i}(k),
  \end{cases}
  \label{eq1}
  \end{equation}
where $x_i(k)\in\mathbb{R}^{n_x}, u_i(k)\in\mathbb{R}^{n_u}$ and $y_i(k)\in\mathbb{R}^{n_y}$ represent state variables, control inputs and measurable output, respectively. The functions $f_i\left(x_{i}(k)\right)$, $I_{i}\left(x_{i}(k)\right), h_{i}\left(x_{i}(k)\right)$, and $F_{i}\left(x_{i}(k)\right)$ are the functions of $x_{i}(k)$ with $f_{i}(0)=0$, $I_{i}(0)=0, h_{i}(0)=0$, and $F_{i}(0)=0$ and $G_i$’s are known matrces. A process uncertainty is denoted by $\omega_i(k)\in\mathbb{R}^{n_\omega}$, and $v_{i}(k)\in\mathbb{R}^{n_v}$ as a measurement noise which are assumed to be confined to specified ellipsoidal sets.

\begin{defi}
An ellipsoidal set has the form $\mathcal{X} \triangleq\{\zeta: \zeta=$ $c+\Xi z,\|z\| \leq 1\}$, where $c \in \mathbb{R}^{n_{x}}$ is the center and $\Xi \in \mathbb{R}^{n_{x} \times m}$ with $\operatorname{rank}(\Xi)=m \leq n_{x}$ is its shape matrix. Assume that $\Xi$ is a lower triangular matrix whose diagonal elements all are positive. According to a Cholesky factorization, it can be seen that $P=\Xi \Xi^{\mathrm{T}}>0$ and $z^{\mathrm{T}} z=(\zeta-c)^{\mathrm{T}} P^{-1}(\zeta-c) \leq 1$. Consequently, the ellipsoidal set can also be represented as $\mathcal{X} \triangleq\left\{\zeta:(\zeta-c)^{\mathrm{T}} P^{-1}(\zeta-c) \leq 1\right\}$. The size of the ellipsoid is dependent on the squares shape matrix $P$ and can be calculated as $\operatorname{Tr}(P)$, which is the sum of the squared semiaxes lengths \cite{b33}.
\end{defi}
\begin{assumption}
The process noise $w_i(k)$ is UBB, which is assumed to belong to the following specified ellipsoidal sets:
\begin{equation}
    \begin{aligned}
        &\mathcal{W}_{i}(k) \triangleq\left\{w_{i}(k): {w_{i}}(k)^{\mathrm{T}} Q_{i}(k)^{-1} w_{i}(k) \leq 1\right\} \\
        &\mathcal{V}_{i}(k) \triangleq\left\{v_{i}(k): {v_{i}^{\mathrm{T}}(k)} R_{i}^{-1}(k) v_{i}(k) \leq 1\right\},
    \end{aligned}\label{eq2}
\end{equation}
where $Q_{i}(k) = Q_{i}(k)^T>0$ and $R_{i}(k) = R_{i}(k)^T>0$ are known matrices with compatible dimensions.\label{assump1}
\end{assumption}

The system model for the $i$th agent is presented by fuzzy IF-THEN rules.

Plant Rule $l_{i}:$ IF $\theta_{i,1}(k)$ is $\mu_{l_{i},1}$ and $\theta_{i,2}(k)$ is $\mu_{l_{i},2} \ldots$ and $\theta_{i,q}(k)$ is $\mu_{l_{i},q}$, THEN
\begin{equation}
    \begin{cases}
        x_{i}(k+1)=A_{l_{i}} x_{i}(k)+B_{l_{i}} u_{i}(k)+M_{l_{i}}\omega_i(k)\\
        y_{i}(k)=C_{l_{i}}x_{i}(k)+D_{l{i}}v_{i}(k),
    \end{cases}\label{eq8}
\end{equation}
where $l_{i}=1, \ldots, r$ ($r$ stands for the total number of plant IF-THEN rules), $\mu_{l_{i},1}, \ldots, \mu_{l_{i},q}$ are fuzzy sets, $\theta_{i}(k)=$ $\left[\theta_{i,1}^{T}(k) \theta_{i,2}^{T}(k) \cdots \theta_{i,q}^{T}(t)\right]^{T}$ denotes the premise variable, $A_{l_{i}}, B_{l_{i}}, M_{l_{i}}, C_{l_{i}}$ and $D_{l_{i}}$ are the system matrices with appropriate dimensions. The above-mentioned system can be inferred as follows:
\begin{equation}
\small{
    \left\{\begin{aligned}
    {x}_{i}(k+1)=&\sum_{l_{i}=1}^{r} g_{l_{i}}\left(\theta_{i}(k)\right)A_{l_{i}} x_{i}(k)\\
    &+\sum_{l_{i}=1}^{r} g_{l_{i}}\left(\theta_{i}(k)\right)B_{l_{i}} u_{i}(k)\\
    &+\sum_{l_{i}=1}^{r} g_{l_{i}}\left(\theta_{i}(k)\right)M_{l{i}}\omega_{i}(k)\\
    y_{i}(k)=&\sum_{l_{i}=1}^{r} g_{l_{i}}\left(\theta_{i}(k)\right)C{l_{i}}x_{i}(k)\\
    &+\sum_{l_{i}=1}^{r} g_{l_{i}}\left(\theta_{i}(k)\right)D_{l{i}}v_{i}(k)
    \end{aligned}\right.,\label{eq9}
    }
\end{equation}
where $g_{l_{i}}\left(\theta_{i}(k)\right)=\psi_{l_{i}}\left(\theta_{i}(k)\right) / \sum_{l_{i}=1}^{r} \psi_{l_{i}}\left(\theta_{i}(k)\right)$ is the normalized weight for each rule with $\psi_{l_{i}}\left(\theta_{i}(k)\right)=$ $\Pi_{v=1}^{q} \mu_{l_{i} v}\left(\theta_{i v}(k)\right) \geqslant 0$ and $\sum_{l_{i}=1}^{r} g_{l_{i}}\left(\theta_{i}(k)\right)=1$, where $\mu_{l_{i} v}\left(\theta_{i v}(k)\right)$ is the grade of membership of $\theta_{i v}(k)$ in $\mu_{l_{i} q}$.

By considering fuzzy model as an interpolation of $r$ linear systems through the membership function $g_{l_{i}}\left(\theta_{i}(k)\right)$, we can approximate the nonlinear system. Therefore, the nonlinear multi-agent system can be described as
\begin{equation}
\small{
    \left\{\begin{aligned}
x_{i}(k+1)=& f_i(x_i(k))+G_iu_i(k)+I_{i}(x_{i})\omega_i(k) \\
=& \sum_{l_{i}=1}^{r} g_{l_{i}}\left(\theta_{i}(k)\right) A_{l_{i}} x_{i}(k)+\Delta f_i\left(x_{i}(k)\right)\\&+\sum_{l_{i}=1}^{r} g_{l_{i}}\left(\theta_{i}(k)\right)B_{l_{i}} u_{i}(k)\\
&+\sum_{l_{i}=1}^{r} g_{l_{i}}\left(\theta_{i}(k)\right)M_{l{i}}\omega_{i}(k)\\
&+\Delta I_i\left(x_{i}(k)\right)\omega_{i}(k) \\
y_{i}(k)=&\sum_{l_{i}=1}^{r} g_{l_{i}}\left(\theta_{i}(k)\right)C_{l_{i}}x_{i}(k)+\Delta h_i\left(x_{i}(k)\right)\\
    &+\sum_{l_{i}=1}^{r} g_{l_{i}}\left(\theta_{i}(k)\right)D_{l{i}}v_{i}(k)\\
    &+\Delta F_i\left(x_{i}(k)\right)v_{i}(k),
\end{aligned}\right.
}\label{eq10}
\end{equation}
where
\begin{equation}
\small{
    \left\{\begin{aligned}
\Delta f_{i}\left(x_{i}(k)\right)=&f_{i}\left(x_{i}(k)\right)-\sum_{l_i=1}^{r} g_{l_{i}}\left(\theta_{i}(k)\right) A_{l_{i}} x_{i}(k) \\
\Delta I_{i}\left(x_{i}(k)\right)=&I_{i}\left(x_{i}(k)\right)-\sum_{l_i=1}^{r} g_{l_{i}}\left(\theta_{i}(k)\right) M_{l_{i}} \\
\Delta h_{i}\left(x_{i}(k)\right)=&h_{i}\left(x_{i}(k)\right)-\sum_{l_i=1}^{r} g_{l_{i}}\left(\theta_{i}(k)\right) C_{l_{i}} x_{i}(k) \\
\Delta F_{i}\left(x_{i}(k)\right)=&F_{i}\left(x_{i}(k)\right)-\sum_{l_i=1}^{r} g_{l_{i}}\left(\theta_{i}(k)\right) D_{l_{i}} \\
\end{aligned}\right.
}
\end{equation}
denote the approximation (or interpolation) errors between the nonlinear system and the fuzzy model.
\begin{assumption}
According to \cite{b27}, we assume
\begin{equation}
\left\{\begin{aligned}
    &\Delta f_i\left(x_{i}(k)\right)=H_{i,1} \Delta_{i,1} E_{i,1} x_{i}(k)\\
    &\Delta I_i\left(x_{i}(k)\right)=H_{i,2} \Delta_{i,2} E_{i,2}\\
    &\Delta h_i\left(x_{i}(k)\right)=H_{i,3} \Delta_{i,3} E_{i,3} x_{i}(k)\\
    &\Delta F_i\left(x_{i}(k)\right)=H_{i,4} \Delta_{i,4} E_{i,4},
    \label{eq12}
    \end{aligned}\right.
\end{equation}
where $H_{i}$ and $E_{i}$ are known matrices, and $\Delta_{i}$ is unknown but bounded with $\left\|\Delta_{i}\right\| \leq 1$.\label{assump2}
\end{assumption}

We are interested in constructing the fuzzy-based leader following consensus protocol, which utilizes the estimated state instead of the full system state. First, consider the leader agent's dynamic by the following IF-THEN rules.\\
Plant Rule $l_{i}:$ IF $\theta_{i,1}(k)$ is $\mu_{l_{i},1}$ and $\theta_{i,2}(k)$ is $\mu_{l_{i},2} \ldots$ and $\theta_{i,q}(k)$ is $\mu_{l_{i},q}$, THEN
\begin{equation}
    x^{l}(k+1)=A_{l_{i}}^l x^l(k),\label{eq11}
\end{equation}
where $x^{l}(k) \in \mathbb{R}^{n_{x}}$ is the state of the leader, and $A_{l_{i}}^l$ are the system matrices with appropriate dimensions. It is assumed that the leader's dynamics are not subject to UBB process noise. The abovementioned system can be inferred as follows:
\begin{equation}
    x^{l}(k+1)=\sum_{l_{i}=1}^{r} g_{l_{i}}\left(\theta_{i}(k)\right)A_{l_{i}}^l x^l(k).\label{ext}
\end{equation}
\begin{assumption}
The initial states $x_{i}(0)$ and $x^{l}(0)$ are assumed to belong to a given ellipsoid
\begin{equation}
\small{
\begin{aligned}
    \mathcal{X}_{i}(0 \mid 0) \triangleq\left\{x\right._{i}(0):
    &\left(x_{i}(0)-\hat{x}_{i}(0 \mid 0)\right)^{\mathrm{T}} P_{i}(0 \mid 0)^{-1}\\
    &\times\left.\left(x_{i}(0)-\hat{x}_{i} (0\mid 0)\right) \leq 1\right\}\\
    \mathcal{U}_{i}(0) \triangleq\left\{x\right._{i}(0):
    &\left(x_{i}(0)-x^{l}(0)\right)^{\mathrm{T}} U_{i}(0)^{-1}\\
    &\times\left.\left(x_{i}(0)-x^{l} (0)\right) \leq 1\right\},
\end{aligned}
}
\end{equation}
where $\hat{x}_{i}(0 \mid 0)$ is the given estimate of $x_{i}(0)$, and $P_{i} (0\mid 0)=P_{i}(0 \mid 0)^{\mathrm{T}} \succ 0$ and $U_{i} (0\mid 0)=U_{i}(0 \mid 0)^{\mathrm{T}} \succ 0$ are known matrices.\label{assump3}
\end{assumption}

In this paper, we consider two kinds of attacks on the system.
\subsection{False Data Injection Attacks}
The original data packets are replaced by false ones when they are transferred from controllers to actuators or from another agent via communication channels.

\begin{equation}
    \begin{cases}
        u_i^c(k)=u_i(k)+u_i^a(k)\\
        %y_i^c(k)=y_i(k)+y_i^a(k)\\
        \Bar{x}_j^c(k)=\Bar{x}_j(k)+\phi_j^a(k)\Bar{x}_j^a(k),\label{eq3}
    \end{cases}
\end{equation}
where $u_i(k)\in\mathbb{R}^{m_i}, u_i^a(k)$ and $u_i^c(k)$ are the uncompromised control input, unknown false data injected to the actuator of agent $i$, and the compromised input available to agent $i$. Moreover, $\Bar{x}_j^c$ is the corrupted neighbouring data, and in the presence of an attack on neighbouring channel $\phi_j^a$ is 1", otherwise it is "0".

\subsection{Replay Attacks}
\label{sub2-a}
A successful replay attack does not need a priori knowledge of the system components. It is assumed that the attacker can record sensor’s measurement data from $k_i$ untill $k_r$ with the window size $\tau = k_r-k_i$ in the first phase. Then, in the second phase, the attacker replays the recorded data to the system from $k = k_r + d$ untill the end of the attack at $k = k_f$, where $d$ is the delay between the recording time and replaying time. We model this attack according to \cite{b8} as
\begin{equation}
    a^{y_i}(k)=y_i(k-\tau)-y_i(k).\label{eq5}
\end{equation}
Thus, the sensor’s data affected by the attack is
\begin{equation}
    \tilde{y}_i(k)=y_i(k)+a^{y_i}(k).\label{eq6}
\end{equation}
We propose a distributed attack detector to detect the aforementioned types of attacks. The modules are tasked to detect attacks as well as ensure that the desired control specifications are satisfied. Also, the method can recover the system performance and mitigate the effects of the attacks. The structure of the system with the detector is shown in Fig. \ref{fig8}.
\begin{figure}[t]
    \centering
    \includegraphics[width=\columnwidth]{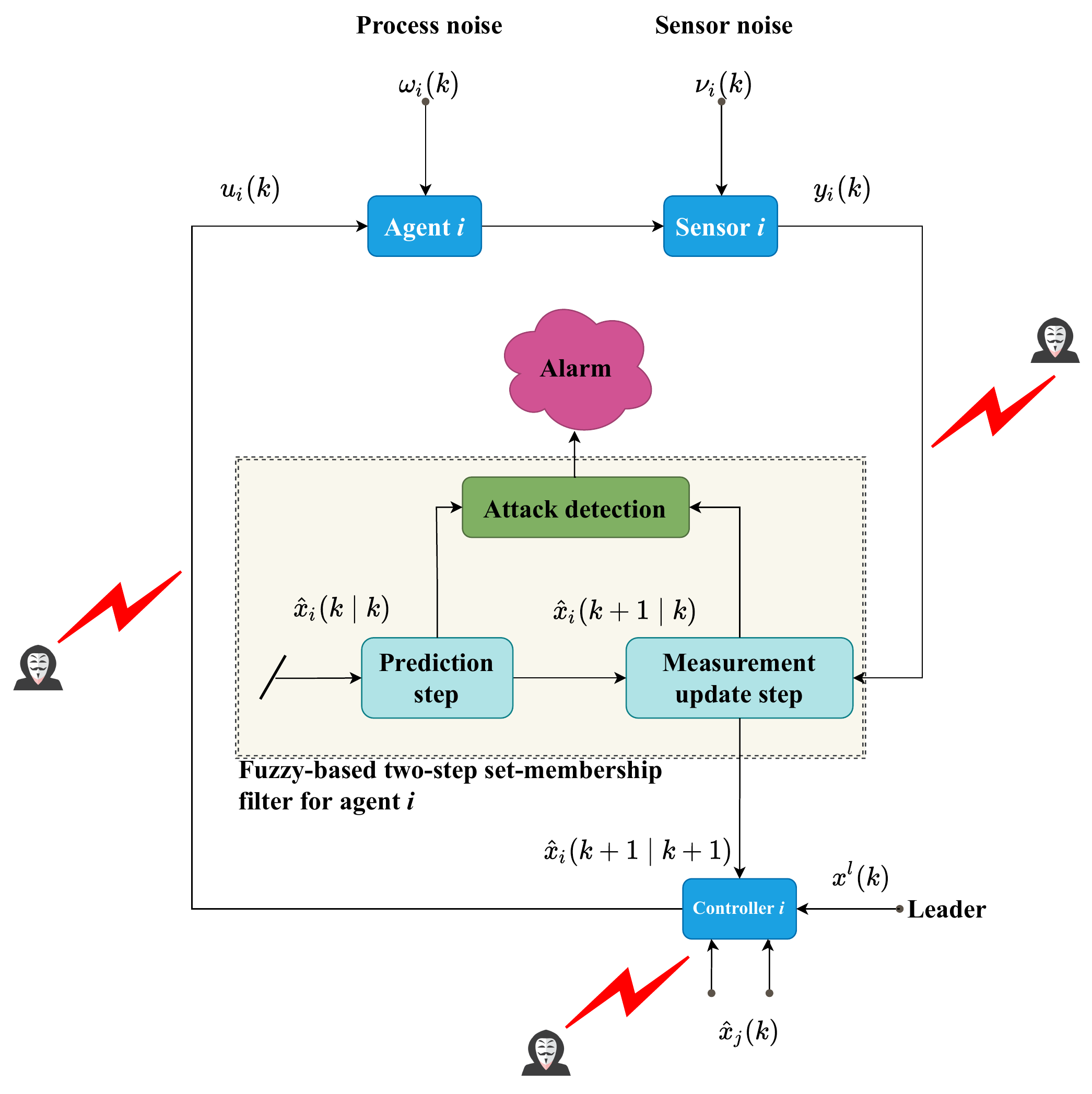}
    \caption{The structure of a leader-following MAS with a fuzzy-based set-membership filtering detection method.}
    \label{fig8}
\end{figure}

\section{Consensus Protocol and Fuzzy-Based Two-Step Set-Membership Estimation Method}
\subsection{Prediction Step}
First, the prediction filter is considered in the form of\\
Plant Rule $l_{i}:$ IF ${\hat{\theta}}_{i,1}(k)$ is $\mu_{l_{i},1}$ and $\hat{\theta}_{i,2}(k)$ is $\mu_{l_{i},2} \ldots$ and $\hat{\theta}_{i,q}(k)$ is $\mu_{l_{i},q}$, THEN
\begin{equation}
    \hat{x}_{i}(k+1 \mid k)=\hat{A}_{l_i} \hat{x}_{i}(k \mid k),\label{eq13}
\end{equation}
where $\hat{x}_{i}(k \mid k)$ is the estimation of the state $x_{i}(k)$, $\hat{A}_{l_i}$ is the fuzzy filter parameter to be determined and $\hat{\theta}_{i}(k)=\left\{\hat{\theta}_{i,1}(k), \hat{\theta}_{i,2}(k), \ldots, \hat{\theta}_{i,q}(k)\right\}$ are premise variables, which maybe functions of the state estimates. The overall fuzzy filter can be written from \eqref{eq13} as \cite{b25}, \cite{b26}
\begin{equation}
\small{
\begin{aligned}
\hat{x}_{i}(k+1 \mid k)=&\sum_{l_{i}=1}^{r} g_{l_{i}}\left(\hat{\theta}_{i}(k)\right)\hat{A}_{l_i} \hat{x}_{i}(k \mid k).
\end{aligned}
}\label{eq14}
\end{equation}
For the given state estimation ellipsoid set $\mathcal{X}_{i}(k \mid k)$ with the center $\hat{x}_{i}(k \mid k)$ and the shape matrix $\Xi_{i}(k \mid k)$, the real state $x_{i}(k)$ can be described by
\begin{equation}
    x_{i}(k)=\hat{x}_{i}(k \mid k)+\Xi_{i}(k \mid k) z_{i}.\label{eq15}
\end{equation}
Then, our goal is to obtain the prediction ellipsoid set
\begin{equation}
\small{
    \begin{aligned}
&\mathcal{X}_{i}(k+1 \mid k) \triangleq\left\{x_{i}(k+1):\right. \\
&\left.\left(x_{i}(k+1)-\hat{x}_{i}(k+1 \mid k)\right)^{\mathrm{T}} P_{i}^{-1}(k+1 \mid k)\right.\\
&\times\left.\left(x_{i}(k+1)-\hat{x}_{i}(k+1 \mid k)\right) \leq 1\right\}.
\end{aligned}
}\label{eq16}
\end{equation}
Note that the state $x_{i}(k+1)$ belongs to such an ellipsoid set for any value of the system noises in their specified sets.
\subsection{Measurement Update Step}
\label{sub3-b}
The update based on the current measurement is considered for the system \eqref{eq10}, which is in the form of\\
Plant Rule $l_{i}:$ IF $\hat{\theta}_{i,1}(k)$ is $\mu_{l_{i},1}$ and $\hat{\theta}_{i,2}(k)$ is $\mu_{l_{i},2} \ldots$ and $\hat{\theta}_{i,q}(k)$ is $\mu_{l_{i},q}$, THEN
\begin{equation}
\resizebox{\hsize}{!}{$
\begin{aligned}
\hat{x}_{i}(k+1\mid k+1)=&\hat{x}_{i} (k+1\mid k)+L_{l_i}\left(y_{i}(k+1)\right.\left.-\hat{y}_{i}(k+1 \mid k)\right),\label{eq17}
\end{aligned}
$}
\end{equation}
where $L_{l_i}$ is the filter parameter to be determined.
The overall fuzzy update can be written from \eqref{eq17} as
\begin{equation}
\small{
\begin{aligned}
\hat{x}_{i}(k+1\mid k+1)=&\hat{x}_{i} (k+1\mid k)+\sum_{l_{i}=1}^{r} g_{l_{i}}\left(\hat{\theta}_{i}(k)\right)L_{l_i}\\
&\times\left(y_{i}(k+1)-\hat{y}_{i}(k+1 \mid k)\right).\label{eq18}
\end{aligned}
}
\end{equation}
According to the prediction ellipsoid set $\mathcal{X}_{i}(k+1 \mid k)$ given by \eqref{eq16}, the state $x_{i}(k+1)$ can be written as
\begin{equation}
x_{i}(k+1)=\hat{x}_{i}(k+1 \mid k)+\Xi_{i}(k+1 \mid k) z_{i}.\label{eq19}
\end{equation}
Our objective is to update this prediction set with the one yielding from the current measurement $y_{i}(k+1)$. In other words, we look for an updated ellipsoid set $\mathcal{X}_{i}(k+1 \mid k+1)$ with the center $\hat{x}_{i}(k+1 \mid k+1)$ and the shape matrix $\Xi_{i}(k+1 \mid k+1)$ for the state $x_{i}(k+1)$, given by the current measurement information at the time instant $k+1$. Thus, the updated ellipsoid set should satisfy the condition
\begin{equation}
\small{
    \begin{aligned}
        (x_{i}(k+1)&-\hat{x}_{i}(k+1\mid k+1))^{\mathrm{T}} P_{i}^{-1}(k+1 \mid k+1)\\
        &\times\left(x_{i}(k+1)-\hat{x}_{i}(k+1 \mid k+1)\right) \leq 1,\label{eq20}
    \end{aligned}
    }
\end{equation}
whenever the output constraint
\begin{equation}
\small{
\begin{aligned}
    y_{i}(k+1)=&\sum_{l_{i}=1}^{r} g_{l_{i}}\left({\theta}_{i}(k)\right)C_{l_{i}}\left(\hat{x}_{i}(k+1 \mid k)+\Xi_{i}(k+1 \mid k) z_{i}\right)\\
    &+H_{i,3} \Delta_{i,3} E_{i,3}\left(\hat{x}_{i}(k+1\mid k)+\Xi_i(k+1\mid k)z_{i}\right)\\
    &+\left(\sum_{l_{i}=1}^{r} g_{l_{i}}\left({\theta}_{i}(k)\right)D_{l_{i}}+H_{i,4} \Delta_{i,4} E_{i,4}\right)v_{i}(k+1)\label{eq21}
    \end{aligned}
    }
\end{equation}
holds for some $\|z_{i}\| \leq 1$.
\subsection{Leader Following Consensus Protocol}
The distributed observer-based leader following consensus protocol \cite{b32} is
%Plant Rule $l_{i}:$ IF $\hat{\theta}_{i,1}(k)$ is $\mu_{l_{i},1}$ and $\hat{\theta}_{i,2}(k)$ is $\mu_{l_{i},2} \ldots$ and $\hat{\theta}_{i,q}(k)$ is $\mu_{l_{i},q}$, THEN
\begin{equation}
\small{
\begin{aligned}
    u_{i}(k)=K_{l_{i}}\Bigg(\sum_{j \in \mathcal{N}_{i}} &a_{i j}\left(\hat{x}_{i}(k \mid k)-\hat{x}_{j}(k \mid k)\right)\\
    &+\lambda_{i}\left(\hat{x}_{i}(k \mid k)-x^{l}(k)\right)\Bigg),
    \end{aligned}
    }\label{eq22}
\end{equation}
where $K_{l_{i}}$ are constant matrices to be designed, $a_{i j}$ is a nonnegative element of the weighted adjacency matrix $\mathcal{A}=\left[a_{i j}\right] \in \mathbb{R}^{N \times N}$. The adjacency matrix of the topology is selected as a binary matrix, where $a_{i j}=1$ if follower $i$ can receive information from follower $j$, otherwise $a_{i j}=0$.

The leader-following multi-agent system \eqref{eq1}, \eqref{eq11} achieves set-membership leader-following consensus under protocol \eqref{eq22} and two-step filter \eqref{eq14}, \eqref{eq18}, if the existence of desired gain sequences $K_{l_i}, \hat{A}_{l_{i}}$, and $L_{l_{i}}$ can guarantee that the one step ahead states $x_{i}(k+1), \forall i \in \mathcal{\nu}$ for all the followers reside in a leader state ellipsoid $\mathcal{U}_{i}(k+1)$ always enclosing all the followers' true states, where
\begin{equation}
\small{
\begin{aligned}
&\mathcal{U}_{i}(k+1) \triangleq\left\{x_{i}(k+1):\right.\\
&\left(x_{i}(k+1)-x^{l}(k+1)\right)^{T}U^{-1}(k+1)\\ &\left.\times\left(x_{i}(k+1)-x^{l}(k+1)\right) \leq 1\right\}
\end{aligned}
}\label{eq23}
\end{equation}
For the given leader ellipsoid set $\mathcal{U}_{i}(k)$ with the center $x^{l}(k)$ and the shape matrix $\xi_{i}(k)$, the state $x_{i}(k)$ can be described by
\begin{equation}
    x_{i}(k)=x^{l}(k)+\xi_{i}(k) z_{i}.\label{add17}
\end{equation}
\section{Attack Detection Using Set-Membership Fuzzy Filtering}
The proposed cyberattack detection problem is addressed in this section by developing a set-membership filter. First, we develop the prediction elipsoidal sets based on the leader following consensus protocol \eqref{eq22} and then update the prediction ellipsoid set with the current measurement. Also, we develop the leader elipsoidal set based on the leader following consensus protocol \eqref{eq22}. Finally, convex optimization problems and one algorithm are provided to expose the cyberattack diagnosis scheme.
\subsection{The Prediction Ellipsoid Set Design Based on Leader Following Consensus}
From the system model \eqref{eq10} and \eqref{eq12}, and the filter \eqref{eq14} and \eqref{eq15}, the prediction error $x_{i}(k+1)-\hat{x}_{i}(k+1 \mid k)$ can be written as\\
\begin{comment}
\begin{equation}
\resizebox{0.89\hsize}{!}{$
\begin{aligned}
x_{i}(k+1)&-\hat{x}_{i}(k+1 \mid k)\\
%=& \sum_{l_{i}=1}^{r} h_{l_{i}}\left(\theta_{i}(k)\right) A_{l_{i}} x_{i}(k)
%\\&+\Delta f_i\left(x_{i}(k)\right)+\sum_{l_{i}=1}^{r} h_{l_{i}}\left(\theta_{i}(k)\right)G_{l_{i}} u_{i}(k)\\&+\omega_{i}(k)-\sum_{l_{j}=1}^{r} h_{l_{j}}\left(\hat{\theta}_{i}(k)\right)F_{l_j} \hat{x}_{i}(k \mid k)\\
%&-\sum_{l_{j}=1}^{r} h_{l_{j}}\left(\hat{\theta}_{i}(k)\right)B_{l_{j}} u_{i}(k)\\
=&\sum_{l_{i}=1}^{r} g_{l_{i}}\left(\theta_{i}(k)\right) A_{l_{i}}\hat{x}_{i}(k \mid k)\\&+\sum_{l_{i}=1}^{r} g_{l_{i}}\left(\theta_{i}(k)\right) A_{l_{i}}\Xi_{i}(k \mid k) z_{i}\\
&+H_{i,1} \Delta_{i,1} E_{i,1} \hat{x}_{i}(k \mid k)+H_{i,1} \Delta_{i,1} E_{i,1}\Xi_{i}(k \mid k) z_{i}\\
&+\sum_{l_{i}=1}^{r}g_{l_{i}}\left(\theta_{i}(k)\right)B_{l_{i}}
K_{l_{i}}\\&\times\Bigg(\sum_{j \in \mathcal{N}_{i}} a_{i j}\left(\hat{x}_{i}(k\mid k)-\hat{x}_{j}(k\mid k)\right)\\&\quad\quad+\lambda_{i}\left(\hat{x}_{i}(k \mid k)-x^{l}(k)\right)\Bigg)\\
&+\sum_{l_{i}=1}^{r} g_{l_{i}}\left(\theta_{i}(k)\right)M_{l{i}}\omega_{i}(k)\\
&+H_{i,2} \Delta_{i,2} E_{i,2}\omega_{i}(k)-\sum_{j_{i}=1}^{r} g_{j_{i}}\left(\hat{\theta}_{i}(k)\right)F_{j_i} \hat{x}_{i}(k \mid k)\\&-\sum_{j_{i}=1}^{r} g_{j_{i}}\left(\hat{\theta}_{i}(k)\right)\hat{B}_{j_{i}}K_{j_{i}}\\
&\Bigg(\sum_{j \in \mathcal{N}_{i}} a_{i j}\left(\hat{x}_{i}(k\mid k)-\hat{x}_{j}(k\mid k)\right)\\&\quad+\lambda_{i}\left(\hat{x}_{i}(k \mid k)-x^{l}(k)\right)\Bigg)
\end{aligned}
$}
\end{equation}
\end{comment}
%we achieve
\begin{equation}
\small{
\begin{aligned}
x_{i}(k&+1)-\hat{x}_{i}(k+1 \mid k)\\
=&\left(\sum_{l_{i}=1}^{r} g_{l_{i}}\left(\theta_{i}(k)\right) A_{l_{i}}-\sum_{j_{i}=1}^{r} g_{j_{i}}\left(\hat{\theta}_{i}(k)\right)\hat{A}_{j_i}\right)\hat{x}_{i}(k \mid k)\\&+\sum_{l_{i}=1}^{r} g_{l_{i}}\left(\theta_{i}(k)\right) A_{l_{i}}\Xi_{i}(k \mid k) z_{i}+\sum_{l_{i}=1}^{r} g_{l_{i}}\left(\theta_{i}(k)\right)B_{l_{i}}
K_{l_{i}}\\&\times\left(\sum_{j \in \mathcal{N}_{i}} a_{i j}\left(\hat{x}_{i}(k\mid k)-\hat{x}_{j}(k\mid k)\right)+\lambda_{i}\left(\hat{x}_{i}(k \mid k)-x^{l}(k)\right)\right)\\
&+\sum_{l_{i}=1}^{r} g_{l_{i}}\left(\theta_{i}(k)\right)M_{l{i}}\omega_{i}(k)+H_{i,1} q_{i,1}+H_{i,1} q_{i,2}+H_{i,2} q_{i,3},
\end{aligned}
}\label{eq24}
\end{equation}
where
\begin{equation}
    \begin{aligned}
    &q_{i,1}=\Delta_{i,1} E_{i,1} \hat{x}_{i}(k \mid k)\\
    &q_{i,2}=\Delta_{i,1} E_{i,1} \Xi_{i}(k \mid k) z_{i}\\
    &q_{i,3}=\Delta_{i,2} E_{i,2}\omega_{i}(k).
    \end{aligned}
    \label{eq25}
\end{equation}
Denoting
\begin{equation}
    \eta_{i,1}(k)=\left[\begin{array}{llllll}
1 
&z_{i} 
&\omega_{i}(k)
&q_{i,1}
&q_{i,2}
&q_{i,3}
\end{array}\right]^T,\label{eq27}
\end{equation}
and considering the fact that $\sum_{l_{i}=1}^{r} g_{l_{i}}\left(\theta_{i}(k)\right)=1$, we can write \eqref{eq24} in a compact form as
\begin{equation}
\small{
\begin{aligned}
    x_{i}(k+1)&-\hat{x}_{i}(k+1 \mid k)\\=&\sum_{l_{i}=1}^{r} g_{l_{i}}\left(\theta_{i}(k)\right)\sum_{j_{i}=1}^{r} g_{j_{i}}\left(\hat{\theta}_{i}(k)\right)\\
    &\times%\sum_{g_{i}=1}^{r} h_{g_{i}}\left(\theta_{i}(k)\right)
    \Gamma_{i,1,l_{i} j_{i}} \eta_{i,1}(k).
    \end{aligned}
    }\label{eq28}
\end{equation}
By denoting
\small{
\begin{equation}
\begin{aligned}
    P_{i,1,l_{i},j_{i}}=&\left(A_{l_{i}}-\hat{A}_{j_{i}}\right)\hat{x}_{i}(k\mid k)-B_{l_{i}}K_{l_{i}}\lambda_{i}x^{l}(k)\\&+B_{l_{i}}K_{l_{i}}\sum_{j=1}^{N}\tilde{l}_{i j}\hat{x}_{j}(k\mid k),
    \end{aligned}
\end{equation}
}
we have
\begin{equation}
\small{
\begin{aligned}
\Gamma_{i,1,l_{i} j_{i}}=\left[P_{i,1,l_{i},j_{i}}
\quad A_{l_{i}}\Xi_{i}(k\mid k) \quad M_{l_{i}}\quad H_{i,1}
\quad H_{i,1}\quad H_{i,2}\right],
\end{aligned}
}\label{eq29}
\end{equation}
where $\tilde{\mathcal{L}}=\mathcal{L}+\Lambda=\left[\tilde{l}_{i j}\right]_{N \times N}$ and $\Lambda=\operatorname{diag}\left\{\lambda_{1}, \lambda_{2}, \ldots, \lambda_{N}\right\}$.

According to \eqref{eq28}, we can write
\begin{equation}
\small{
    \begin{aligned}
\left(x_{i}(k+1)\right.&\left.-\hat{x}_{i}(k+1 \mid k)\right)^{T} P_{i}^{-1}(k+1 \mid k)\\
&\times\left(x_{i}(k+1)-\hat{x}_{i}(k+1 \mid k)\right) \\
=&\sum_{l_{i}=1}^{r} g_{l_{i}}\left(\theta_{i}(k)\right) \sum_{j_{i}=1}^{r} g_{j_{i}}\left(\hat{\theta}_{i}(k)\right)\\&\times\sum_{m_{i}=1}^{r} g_{m_{i}}\left(\theta_{i}(k)\right) \sum_{n_{i}=1}^{r} g_{n_{i}}\left(\hat{\theta}_{i}(k)\right)\\&\times\eta_{i,1}^{T}(k) \Gamma_{i,1,l_{i} j_{i}}^{T}P_{i}^{-1}(k+1 \mid k)\\&\times\Gamma_{i,1,m_{i} n_{i}}\eta_{i,1}(k).
\end{aligned}
}
\end{equation}
Therefore, we can achieve
\begin{equation}
\small{
    \begin{aligned}
\left(x_{i}(k+1)\right.&\left.-\hat{x}_{i}(k+1 \mid k)\right)^{T} P_{i}^{-1}(k+1 \mid k)\\
&\times\left(x_{i}(k+1)-\hat{x}_{i}(k+1 \mid k)\right) \\
\leq&\sum_{l_{i}=1}^{r} g_{l_{i}}\left(\theta_{i}(k)\right) \sum_{j_{i}=1}^{r} g_{j_{i}}\left(\hat{\theta}_{i}(k)\right)\\&\times\eta_{i,1}^{T}(k) \Gamma_{i,1,l_{i} j_{i}}^{T}P_{i}^{-1}(k+1 \mid k)\\&\times\Gamma_{i,1,l_{i} j_{i}}\eta_{i,1}(k).
\end{aligned}
}\label{eq102}
\end{equation}
The condition in \eqref{eq16} can be written as
\begin{equation}
\small{
\begin{aligned}
    \eta_{i,1}^{T}(k)\Bigg[&\sum_{l_{i}=1}^{r} g_{l_{i}}\left(\theta_{i}(k)\right) \sum_{j_{i}=1}^{r} g_{j_{i}}\left(\hat{\theta}_{i}(k)\right)\\&\times\Gamma_{i,1,l_{i} j_{i}}^{T}P_{i}^{-1}(k+1 \mid k)\Gamma_{i,1,l_{i} j_{i}}\\&-\operatorname{diag}\{1,0,0,0,0,0\}\Bigg] \eta_{i,1}(k) \leq 0.
    \end{aligned}
    }\label{eq30}
\end{equation}
With $\left\|\Delta_{i}\right\| \leq 1$, we can infer from \eqref{eq25} that
\begin{equation}
\left\{\begin{array}{l}
q_{i,1}^{T} q_{i,1}-\hat{x}_{i}^{T}(k \mid k) E_{i,1}^{T} E_{i,1} \hat{x}_{i}(k \mid k) \leq 0 \\
q_{i,2}^{T} q_{i,2}-z_{i}^{T} \Xi_{i}^{T}(k \mid k) E_{i,1}^{T} E_{i,1} \Xi_{i}(k \mid k) z_{i} \leq 0 \\
q_{i,3}^{T} q_{i,3}-\omega_{i}^{T}(k)E_{i,2}^{T} E_{i,2} \omega_{i}(k) \leq 0
\end{array}\right..\label{eq31}
\end{equation}
From \eqref{eq2}, \eqref{eq15} and \eqref{eq31}, the unknown variables $z_{i}, \omega_{i}(k), q_{i,1}, q_{i,2}$ and $q_{i,3}$ satisfy the following constraints:
\begin{equation}
    \begin{aligned}
&\|z_{i}\| \leq 1 \\
&w_{i}^{T}(k) Q_{i}^{-1}(k) w_{i}(k) \leq 1 \\
&q_{i,1}^{T} q_{i,1}-\hat{x}_{i}^{T}(k \mid k) E_{i,1}^{T} E_{i,1} \hat{x}_{i}(k \mid k) \leq 0 \\
&q_{i,2}^{T} q_{i,2}-z_{i}^{T} \Xi_{i}^{T}(k \mid k) E_{i,1}^{T} E_{i,1} \Xi_{i}(k \mid k) z_{i} \leq 0 \\
&q_{i,3}^{T} q_{i,3}-\omega_{i}^{T}(k)E_{i,2}^{T} E_{i,2} \omega_{i}(k) \leq 0,
\end{aligned}
\end{equation}
which can be written in $\eta_{i,1}(k)$ as
\begin{equation}
\begin{aligned}
&\eta_{i,1}^{T}(k) \operatorname{diag}\{-1, I, 0,0,0,0\} \eta_{i,1}(k) \leq 0 \\
&\eta_{i,1}^{T}(k) \operatorname{diag}\{-1,0, Q_{i}^{-1}(k), 0,0,0\} \eta_{i,1}(k) \leq 0 \\
&\eta_{i,1}^{T}(k)\operatorname{diag}\{-\hat{x}_{i}^{T}(k \mid k) E_{i,1}^{T} E_{i,1}\hat{x}_{i}(k \mid k),\\&\quad\quad\quad\quad\quad\quad\quad\quad\quad\quad0,0, I, 0,0\}\eta_{i,1}(k) \leq 0\\
&\eta_{i,1}^{T}(k)\operatorname{diag}\{0,-\Xi_{i}^{T}(k \mid k) E_{i,1}^{T} E_{i,1} \Xi_{i}(k \mid k),\\&\quad\quad\quad\quad\quad\quad\quad\quad\quad\quad0,0,I,0\}\eta_{i,1}(k) \leq 0.\\
&\eta_{i,1}^{T}(k)\operatorname{diag}\{0,0,-E_{i,2}^{T}E_{i,2},0,0,I\}\eta_{i,1}(k) \leq 0.
\end{aligned}\label{eq33}
\end{equation}
Applying S-procedure \cite{b28} to \eqref{eq30} and \eqref{eq33}, we can conclude that the inequality \eqref{eq30} holds if there exist nonnegative scalars $\tau_{i, 1}(k)$, $\tau_{i, 2}(k)$, $\tau_{i, 3}(k), \tau_{i, 4}(k)$ and $\tau_{i, 5}(k)$ such that
\begin{equation}
    \begin{aligned}
&\Gamma_{i,1,l_{i}j_{i}}^{T}P_{i}^{-1}(k+1 \mid k) \Gamma_{i,1,l_{i} j_{i}}-\operatorname{diag}\{1,0,0,0,0,0\}\\&-\tau_{i, 1}(k) \operatorname{diag}\{-1, I, 0,0,0,0\} \\
&-\tau_{i, 2}(k) \operatorname{diag}\{-1,0, Q_{i}^{-1}(k), 0,0,0\} \\
&-\tau_{i, 3}(k)\\
&\quad\times\operatorname{diag}\{-\hat{x}_{i}^{T}(k \mid k) E_{i,1}^{T} E_{i,1}\hat{x}_{i}(k \mid k), 0,0, I, 0,0\} \\
&-\tau_{i, 4}(k)\\
&\quad\times\operatorname{diag}\{0,-\Xi_{i}^{T}(k \mid k) E_{i,1}^{T} E_{i,1} \Xi_{i}(k \mid k),0,0,I,0\}\\
&-\tau_{i, 5}(k)\operatorname{diag}\{0,0,-E_{i,2}^{T}E_{i,2},0,0,I\}\leq 0.
\end{aligned}\label{eq34}
\end{equation}
Inequality \eqref{eq34} can be written in the following compact form
\begin{equation}
    \begin{aligned}
\Gamma_{i,1,l_{i} j_{i}}^{T}&P_{i}^{-1}(k+1 \mid k)\Gamma_{i,1,l_{i} j_{i}}-\operatorname{diag}\{1-\tau_{i,1}(k)\\&-\tau_{i, 2}(k)-\tau_{i, 3}(k)\hat{x}_{i}^{T}(k \mid k) E_{i,1}^{T} E_{i,1}\hat{x}_{i}(k \mid k),\\&\tau_{i, 1}(k)I-\tau_{i, 4}(k)\Xi_{i}^{T}(k \mid k) E_{i,1}^{T} E_{i,1} \Xi_{i}(k \mid k),\\&\tau_{i, 2}(k)Q_{i}^{-1}(k)-\tau_{i, 5}(k)E_{i,2}^{T}E_{i,2},\\
&\tau_{i, 3}(k)I,\tau_{i, 4}(k)I,\tau_{i, 5}(k)I\}\leq 0.
\end{aligned}\label{eq35}
\end{equation}
Finally, denoting
\begin{equation}
    \begin{aligned}
\Theta_{i,1}&(k)\\
=&\operatorname{diag}\{ 1-\tau_{i, 1}(k)-\tau_{i, 2}(k)\\&-\tau_{i, 3}(k)\hat{x}_{i}^{T}(k \mid k) E_{i,1}^{T} E_{i,1}\hat{x}_{i}(k \mid k),\tau_{i, 1}(k)I\\&-\tau_{i, 4}(k)\Xi_{i}^{T}(k \mid k) E_{i,1}^{T} E_{i,1} \Xi_{i}(k \mid k),\\&\tau_{i, 2}(k)Q_{i}^{-1}(k)-\tau_{i, 5}(k)E_{i,2}^{T}E_{i,2},\tau_{i, 3}(k)I,\\
&\tau_{i, 4}(k)I,\tau_{i, 5}(k)I\},
\end{aligned}\label{eq36}
\end{equation}
we can write \eqref{eq35} as
\begin{equation}
    \begin{aligned}
&\Gamma_{i,1,l_{i} j_{i}}^{T}P_{i}^{-1}(k+1 \mid k) \Gamma_{i,1,l_{i} j_{i}}-\Theta_{i,1}(k)\leq 0.
\end{aligned}\label{eq37}
\end{equation}
By using Schur complements, \eqref{eq37} is equivalent to
\begin{equation}
    \left[\begin{array}{cc}
-P_{i}(k+1 \mid k) & \Gamma_{i,1,l_{i} j_{i}} \\
\Gamma_{i,1,l_{i} j_{i}}^{\mathrm{T}} &-\Theta_{i,1}(k)
\end{array}\right] \leq 0.\label{eq38}
\end{equation}

Moreover, from the system model \eqref{eq10}, \eqref{eq11}, and \eqref{add17}, and by considering \eqref{eq25} the error $x_{i}(k+1)-x^{l}(k+1 \mid k)$ can be written as
\begin{equation}
\small{
\begin{aligned}
x_{i}(k&+1)-x^{l}(k+1)\\
=&\sum_{l_{i}=1}^{r} g_{l_{i}}\left(\theta_{i}(k)\right) A_{l_{i}}\left(x^{l}(k)+\xi_{i}(k) z_{i}\right)+\sum_{l_{i}=1}^{r}g_{l_{i}}\left(\theta_{i}(k)\right)B_{l_{i}}
K_{l_{i}}\\&\times\left(\sum_{j \in \mathcal{N}_{i}} a_{i j}\left(\hat{x}_{i}(k\mid k)-\hat{x}_{j}(k\mid k)\right)+\lambda_{i}\left(\hat{x}_{i}(k \mid k)-x^{l}(k)\right)\right)\\
&+\sum_{l_{i}=1}^{r} g_{l_{i}}\left(\theta_{i}(k)\right)M_{l{i}}\omega_{i}(k)-\sum_{l_{i}=1}^{r} g_{l_{i}}\left(\theta_{i}(k)\right)A_{l_{i}}^l x^l(k)\\
&+H_{i,1}q_{i,1}+H_{i,1}q_{i,2}+H_{i,2}q_{i,3}.
\end{aligned}
}\label{eq65}
\end{equation}
Considering the fact that $\sum_{l_{i}=1}^{r} g_{l_{i}}\left(\theta_{i}(k)\right)=1$, we can write \eqref{eq65} in a compact form as
\begin{equation}
\begin{aligned}
    x_{i}(k+1)&-x^{l}(k+1)=\sum_{l_{i}=1}^{r} g_{l_{i}}\left(\theta_{i}(k)\right)\sum_{j_{i}=1}^{r}
    \Gamma_{i,2,l_{i}} \eta_{i,1}(k).
    \end{aligned}\label{eq67}
\end{equation}
By denoting
\begin{equation}
\begin{aligned}
    P_{i,2,l_{i}}=&\left(A_{l_{i}}+A_{l{i}}^{l}-B_{l_{i}}K_{l_{i}}\lambda_{i}\right)x^{l}(k)\\&+B_{l_{i}}K_{l_{i}}\sum_{j=1}^{N}\tilde{l}_{i j}\hat{x}_{j}(k\mid k),
    \end{aligned}
\end{equation}
we have
\begin{equation}
\small{
\begin{aligned}
\Gamma_{i,2,l_{i}}=\left[P_{i,2,l_{i}}
\quad A_{l_{i}}\xi_{i}(k) \quad M_{l_{i}} \quad H_{i,1} \quad H_{i,1} \quad H_{i,2}\right].
\end{aligned}
}\label{eq68}
\end{equation}
According to \eqref{eq67}, we can write
\begin{equation}
\small{
    \begin{aligned}
\left(x_{i}(k+1)\right.&\left.-x^{l}(k+1)\right)^{T} U_{i}^{-1}(k+1)\\
&\times\left(x_{i}(k+1)-x^{l}(k+1)\right) \\
=&\sum_{l_{i}=1}^{r} g_{l_{i}}\left(\theta_{i}(k)\right) \sum_{j_{i}=1}^{r} g_{j_{i}}\left(\theta_{i}(k)\right)\\&\times\eta_{i,1}^{T}(k) \Gamma_{i,2,l_{i}}^{T}U_{i}^{-1}(k+1)\\&\times\Gamma_{i,2,j_{i}}\eta_{i,1}(k).
\end{aligned}
}
\end{equation}
Therefore, we can achieve
\begin{equation}
\small{
    \begin{aligned}
\left(x_{i}(k+1)\right.&\left.-x^{l}(k+1)\right)^{T} U_{i}^{-1}(k+1)\\
&\times\left(x_{i}(k+1)-x^{l}(k+1)\right) \\
\leq&\sum_{l_{i}=1}^{r} g_{l_{i}}\left(\theta_{i}(k)\right) \eta_{i,1}^{T}(k) \Gamma_{i,2,l_{i}}^{T}U_{i}^{-1}(k+1)\\&\times\Gamma_{i,2,l_{i}}\eta_{i,1}(k).
\end{aligned}
}\label{eq103}
\end{equation}
Therefore, the condition in \eqref{eq23} can be written as
\begin{equation}
\begin{aligned}
    \eta_{i,1}^{T}(k)\Bigg[&\sum_{l_{i}=1}^{r} g_{l_{i}}\left(\theta_{i}(k)\right)\Gamma_{i,2,l_{i}}^{T}U_{i}^{-1}(k+1)\Gamma_{i,2,l_{i}}\\&-\operatorname{diag}\{1,0,0,0,0,0\}\Bigg] \eta_{i,1}(k) \leq 0.
    \end{aligned}\label{eq71}
\end{equation}
Applying S-procedure to \eqref{eq33} and \eqref{eq71}, we can conclude that the inequality \eqref{eq71} holds if there exist nonnegative scalars $\tau_{i, 6}(k)$, $\tau_{i, 7}(k)$, $\tau_{i, 8}(k), \tau_{i, 9}(k)$ and $\tau_{i, 10}(k)$ such that
\begin{equation}
    \begin{aligned}
&\Gamma_{i,2,l_{i}}^{T}U_{i}^{-1}(k+1) \Gamma_{i,2,l_{i}}-\operatorname{diag}\{1,0,0,0,0,0\}\\&-\tau_{i, 6}(k) \operatorname{diag}\{-1, I, 0,0,0,0\} \\
&-\tau_{i, 7}(k) \operatorname{diag}\{-1,0, Q_{i}^{-1}(k), 0,0,0\} \\
&-\tau_{i, 8}(k)\operatorname{diag}\{-\hat{x}_{i}^{T}(k \mid k) E_{i,1}^{T} E_{i,1}\hat{x}_{i}(k \mid k), 0,0, I, 0,0\} \\
&-\tau_{i, 9}(k)\operatorname{diag}\{0,-\Xi_{i}^{T}(k \mid k) E_{i,1}^{T} E_{i,1} \Xi_{i}(k \mid k),0,0,I,0\}\\
&-\tau_{i, 10}(k)\operatorname{diag}\{0,0,-E_{i,2}^{T}E_{i,2},0,0,I\}\leq 0.
\end{aligned}\label{eq73}
\end{equation}
Inequality \eqref{eq73} can be written in the following compact form
\begin{equation}
    \begin{aligned}
\Gamma_{i,2,l_{i}}^{T}&U_{i}^{-1}(k+1)\Gamma_{i,2,l_{i}}-\operatorname{diag}\{1-\tau_{i,6}(k)-\tau_{i,7}(k)\\&-\tau_{i,8}(k)\hat{x}_{i}^{T}(k \mid k) E_{i,1}^{T} E_{i,1}\hat{x}_{i}(k \mid k),\tau_{i, 6}(k)I\\&-\tau_{i, 9}(k)\Xi_{i}^{T}(k \mid k) E_{i,1}^{T} E_{i,1} \Xi_{i}(k \mid k),\tau_{i, 7}(k)Q_{i}^{-1}(k)\\&-\tau_{i, 10}(k)E{i,2}^{T}E_{i,2},\tau_{i, 8}(k)I,\tau_{i, 9}(k)I,\\&\tau_{i, 10}(k)I\}\leq 0.
\end{aligned}\label{eq74}
\end{equation}
Finally, denoting
\begin{equation}
\small{
    \begin{aligned}
\Theta_{i,2}&(k)\\
=&\operatorname{diag}\{ 1-\tau_{i, 6}(k)-\tau_{i, 7}(k)\\&-\tau_{i, 8}(k)\hat{x}_{i}^{T}(k \mid k) E_{i,1}^{T} E_{i,1}\hat{x}_{i}(k \mid k),\tau_{i, 6}(k)I\\&-\tau_{i, 9}(k)\Xi_{i}^{T}(k \mid k) E_{i,1}^{T} E_{i,1} \Xi_{i}(k \mid k),\tau_{i, 7}(k)Q_{i}^{-1}(k)\\&-\tau_{i, 10}(k)E_{i,2}^{T}E_{i,2},\tau_{i, 8}(k)I,\tau_{i, 9}(k)I,\tau_{i, 10}(k)I\},
\end{aligned}
}\label{eq75}
\end{equation}
we can write \eqref{eq74} as
\begin{equation}
    \begin{aligned}
&\Gamma_{i,2,l_{i}}^{T}U_{i}^{-1}(k+1) \Gamma_{i,2,l_{i}}-\Theta_{i,2}(k)\leq 0.
\end{aligned}\label{eq76}
\end{equation}
By using Schur complements, \eqref{eq76} is equivalent to
\begin{equation}
    \left[\begin{array}{cc}
-U_{i}(k+1) & \Gamma_{i,2,l_{i}} \\
\Gamma_{i,2,l_{i}}^{\mathrm{T}} &-\Theta_{i,2}(k)
\end{array}\right] \leq 0.\label{eq77}
\end{equation}
\begin{theorem}
Consider the leader following multi-agent system \eqref{eq1}, \eqref{eq11} that satisfies Assumption \ref{assump1}, Assumption \ref{assump2} and Assumption \ref{assump3}. Suppose that the state $x_{i}(k)$ belongs to its state estimation ellipsoid $\left(x_{i}(k)-\hat{x}_{i}(k \mid k)\right)^{\mathrm{T}} P_{i}^{-1}(k \mid k)\left(x_{i}(k)-\hat{x}_{i}(k \mid k)\right) \leq 1$ and leader state ellipsoid $\left(x_{i}(k)-x^{l}(k)\right)^{\mathrm{T}} U_{i}^{-1}(k)\left(x_{i}(k)-x^{l}(k)\right) \leq 1$, then the one-step ahead state $x_{i}(k+1)$ will reside in its state prediction ellipsoid $\left(x_{i}(k+1)-\hat{x}_{i}(k+1 \mid k)\right)^{\mathrm{T}} P_{i}^{-1}(k+1 \mid k)\left(x_{i}(k+1)-\hat{x}_{i}(k+1 \mid k)\right) \leq 1$ as well as leader state ellipsoid $\left(x_{i}(k+1)-x^{l}(k+1)\right)^{\mathrm{T}} U_{i}^{-1}(k+1)\left(x_{i}(k+1)-x^{l}(k+1)\right) \leq 1$, if there exist $P_{i}(k+1 \mid k)>0, U_{i}(k+1)>0, \hat{A}_{l_{i}}, K_{l{i}}, \tau_{i,m}(k) \geq 0,$ for $m=1,\ldots,10,$ such that the linear matrix inequalities (LMI) \eqref{eq38} and \eqref{eq77} hold for all $l_{i},j_{i}=1,\dots,r$.\label{theo1}
\end{theorem}
\begin{proof}
According to the above discussion, if there exist $P_{i}(k+1 \mid k)>0, U_{i}(k+1)>0, \hat{A}_{l_{i}}, K_{l{i}}, \tau_{i,m}(k) \geq 0,$ for $m=1,\ldots,10,$ such that \eqref{eq38} and \eqref{eq77} hold for all $l_{i},j_{i}=1,\dots,r$, then we have
\begin{equation}
\begin{aligned}
    \sum_{l_{i}=1}^{r} &g_{l_{i}}\left(\theta_{i}(k)\right) \sum_{j_{i}=1}^{r} g_{j_{i}}\left(\hat{\theta}_{i}(k)\right)\\&\times\eta_{i,1}^{T}(k)\Gamma_{i,1,l_{i} j_{i}}^{T}P_{i}^{-1}(k+1 \mid k)\Gamma_{i,1,l_{i} j_{i}}\eta_{i,1}(k) \leq 1
    \end{aligned}\label{eq100}
\end{equation}
and
\begin{equation}
\begin{aligned}
    \sum_{l_{i}=1}^{r} g_{l_{i}}\left(\theta_{i}(k)\right)\eta_{i,1}^{T}(k)\Gamma_{i,2,l_{i}}^{T}U_{i}^{-1}(k+1)\Gamma_{i,2,l_{i}}\eta_{i,1}(k) \leq 1.
    \end{aligned}\label{eq101}
\end{equation}
From \eqref{eq102} and \eqref{eq103}, we obtain
\begin{equation}
\resizebox{0.97\hsize}{!}{$
\begin{aligned}
    \left(x_{i}(k+1)-\hat{x}_{i}(k+1 \mid k)\right)^{\mathrm{T}} P_{i}^{-1}(k+1 \mid k)\left(x_{i}(k+1)-\hat{x}_{i}(k+1 \mid k)\right) \leq 1
    \end{aligned}
    $}
\end{equation}
and
\begin{equation}
    \left(x_{i}(k+1)-x^{l}(k+1)\right)^{\mathrm{T}} U_{i}^{-1}(k+1)\left(x_{i}(k+1)-x^{l}(k+1)\right) \leq 1,
\end{equation}
which complete the proof.
\end{proof}
According to the Theorem \ref{theo1} and in order to find the optimal state prediction ellipsoid containing $x_{i}(k+1)$, the convex optimization is performed as
\begin{equation}
\begin{aligned}
\underset{\begin{subarray}{c}
  P_{i}(k+1 \mid k), U_{i}(k+1), \hat{A}_{l{i}}(k), K_{l{i}},\\
  \tau_{i,1}(k), \tau_{i, 2}(k), \tau_{i, 3}(k), \tau_{i, 4}(k), \tau_{i, 5}(k),\\
  \tau_{i,6}(k), \tau_{i, 7}(k), \tau_{i, 8}(k), \tau_{i, 9}(k), \tau_{i, 10}(k)\\
  \end{subarray}}{\operatorname{min}}\operatorname{Tr}\left(T_{i}(k+1 \mid k)\right)
\end{aligned}\label{eq39}
\end{equation}
subject to \eqref{eq38} for all $l_{i},j_{i}=1,\dots,r$ in which the trace of $T_{i}(k+1 \mid k)=\operatorname{diag}\{U_{i}(k+1),P_{i}(k+1 \mid k)\}$ is optimized at each time step in order to find the prediction ellipsoid set with minimal size.
\subsection{Update on Prediction Ellipsoid Set With Current Measurement}
We develop here a scheme to determine the shape matrix $E_{i}(k+1 \mid k+1)$ and the filter gain $L_{i}(k+1)$ with the output constraint \eqref{eq21}.

From the system \eqref{eq1}, the prediction ellipsoid set \eqref{eq19}, and the filter based on the current measurement \eqref{eq19}, the current estimation error $x_{i}(k+1)-\hat{x}_{i}(k+1 \mid k+1)$ can be written as\\
\begin{comment}
\begin{equation}
    \begin{aligned}
x_{i}(k+1)&-\hat{x}_{i}(k+1\mid k+1)\\
=&\hat{x}_{i}(k+1 \mid k)+\Xi_{i}(k+1 \mid k) z_{i}\\
&-\hat{x}_{i} (k+1\mid k)-\sum_{j_{i}=1}^{r} g_{j_{i}}\left(\hat{\theta}_{i}(k)\right)L_{j_i}\\&\times\Bigg[\sum_{l_{i}=1}^{r} g_{l_{i}}\left({\theta}_{i}(k)\right)C_{l_{i}}\Xi_{i}(k+1 \mid k) z_{i}\\
&-H_{i,3} \Delta_{i,3} E_{i,3} \hat{x}_{i}(k+1\mid k)\\
&-H_{i,3} \Delta_{i,3} E_{i,3}\Xi_i(k+1\mid k)z_{i}\\
&-\sum_{l_{i}=1}^{r} g_{l_{i}}\left({\theta}_{i}(k)\right)D_{l_{i}}v_{i}(k+1)\\
&-H_{i,4} \Delta_{i,4} E_{i,4}v_{i}(k+1)\Bigg]
\end{aligned}
\end{equation}
\end{comment}
%we achieve
\begin{equation}
\small{
    \begin{aligned}
x_{i}(k+1)&-\hat{x}_{i}(k+1\mid k+1)\\
=&\sum_{l_{i}=1}^{r} g_{l_{i}}\left(\theta_{i}(k)\right)\sum_{j_{i}=1}^{r} g_{j_{i}}\left(\hat{\theta}_{i}(k)\right)\\
&\times\left[\left(I-L_{j_i}C_{l{i}}\right) \Xi_{i}(k+1 \mid k)z_{i}\right.\\
&\quad\quad-D_{l_{i}}v_{i}(k+1)-L_{j_i}H_{i,3}q_{i,4}\\
&\quad\quad\left.-L_{j_i}H_{i,3}q_{i,5}-L_{j_i}H_{i,4}q_{i,6}\right],
\end{aligned}
}
\end{equation}
where
\begin{equation}
    \begin{aligned}
    &q_{i,4}=\Delta_{i,3} E_{i,3} \hat{x}_{i}(k+1 \mid k)\\
    &q_{i,5}=\Delta_{i,3} E_{i,3} \Xi_{i}(k+1 \mid k) z_{i}\\
    &q_{i,6}=\Delta_{i,4} E_{i,4}v_{i}(k+1).
    \end{aligned}\label{eq42}
\end{equation}
Therefore, we can define
\begin{equation}
    \eta_{i,2}(k+1)=\left[\begin{array}{llllll}
1
&z_{i}
&v_i(k+1)\
&q_{i,4}
&q_{i,5}
&q_{i,6}
\end{array}\right]^T.
\end{equation}
Thus, the above estimation error dynamics can be written in a compact form
\begin{equation}
\begin{aligned}
    x_{i}(k+1)&-\hat{x}_{i}(k+1 \mid k+1)\\
    =&\sum_{l_{i}=1}^{r} g_{l_{i}}\left(\theta_{i}(k)\right)\sum_{j_{i}=1}^{r} g_{j_{i}}\left(\hat{\theta}_{i}(k)\right)\\
    &\times\Gamma_{i,2,l_{i}j_{i}}\eta_{i,2}(k+1).
\end{aligned}\label{eq45}
\end{equation}
By denoting
\begin{equation}
\begin{aligned}
    P_{i,2,l_{i},j_{i}}=\left(I-L_{j_i}C_{l{i}}\right)\Xi_{i}(k+1 \mid k),
    \end{aligned}
\end{equation}
we have
\begin{equation}
\resizebox{0.96\hsize}{!}{$
\begin{aligned}
\Gamma_{i,3,l{i}j_{i}}=\left[0\quad P_{i,2,l_{i},j_{i}}\quad-L_{j_{i}}D_{l{i}}
\quad -L_{j_i}H_{i,3} \quad -L_{j_i}H_{i,3}\quad -L_{j_i}H_{i,4}\right].
\end{aligned}
$}
\end{equation}
Taking \eqref{eq45} into account, we can write
\begin{equation}
\small{
\begin{aligned}
(x_{i}(k+1)&-\hat{x}_{i}(k+1\mid k+1))^{\mathrm{T}} P_{i}^{-1}(k+1 \mid k+1)\\
        &\times\left(x_{i}(k+1)-\hat{x}_{i}(k+1 \mid k+1)\right)\\
        =&\sum_{l_{i}=1}^{r}g_{l_{i}}\left(\theta_{i}(k)\right)\sum_{j_{i}=1}^{r} g_{j_{i}}\left(\hat{\theta}_{i}(k)\right)\\
        &\sum_{m_{i}=1}^{r}g_{m_{i}}\left(\theta_{i}(k)\right)\sum_{n_{i}=1}^{r} g_{n_{i}}\left(\hat{\theta}_{i}(k)\right)\\
        &\eta_{i,2}^{T}(k+1)\Gamma_{i,3,l{i}j_{i}}^{T}\Gamma_{i,3,m{i}n_{i}}\eta_{i,2}(k+1)
        \end{aligned}
        }
\end{equation}
Therefore, we can achieve
\begin{equation}
\begin{aligned}
\left(x_{i}(k+1)\right.&\left.-\hat{x}_{i}(k+1\mid k+1)\right)^{\mathrm{T}} P_{i}^{-1}(k+1 \mid k+1)\\
        &\times\left(x_{i}(k+1)-\hat{x}_{i}(k+1 \mid k+1)\right)\\
        \leq&\sum_{l_{i}=1}^{r}g_{l_{i}}\left(\theta_{i}(k)\right)\sum_{j_{i}=1}^{r} g_{j_{i}}\left(\hat{\theta}_{i}(k)\right)\\
        &\eta_{i,2}^{T}(k+1)\Gamma_{i,3,l{i}j_{i}}^{T}\Gamma_{i,3,l{i}j_{i}}\eta_{i,2}(k+1)
        \end{aligned}\label{eq104}
\end{equation}
Therefore, the condition \eqref{eq20} in Section \ref{sub3-b} can be described as
\begin{equation}
\begin{aligned}
    \eta_{i,2}^{\mathrm{T}}(k+1)\Bigg[&\sum_{l_{i}=1}^{r} g_{l_{i}}\left(\theta_{i}(k)\right)\sum_{j_{i}=1}^{r} g_{j_{i}}\left(\hat{\theta}_{i}(k)\right)\Gamma_{i,3,l_{i}j_{i}}^{\mathrm{T}}\\&\times P_{i}^{-1}(k+1\mid k+1)\Gamma_{i,3,l_{i}j_{i}}\\
    &-\operatorname{diag}\{1,0,0,0,0,0\}\Bigg] \eta_{i,2}(k+1) \leq 0
    \end{aligned}\label{eq44}
\end{equation}
On the other hand, from \eqref{eq2}, \eqref{eq19}, and \eqref{eq42} the unknown variables $z_{i}, v_{i}(k+1), q_{i,4}, q_{i,5},$ and $q_{i,6}$ satisfy the following constraints:
\begin{equation}
    \begin{aligned}
&\|z_{i}\| \leq 1 \\
&v_{i}^{T}(k+1) R_{i}^{-1}(k+1) v_{i}(k+1) \leq 1 \\
&q_{i,4}^{T} q_{i,4}-\hat{x}_{i}^{T}(k+1 \mid k) E_{i,3}^{T} E_{i,3} \hat{x}_{i}(k+1 \mid k) \leq 0 \\
&q_{i,5}^{T} q_{i,5}-z_{i}^{T}\Xi_{i}^{T}(k+1 \mid k)\\
&\quad\quad\quad\quad\quad\quad\quad\quad\times E_{i,3}^{T} E_{i,3} \Xi_{i}(k+1 \mid k) z_{i} \leq 0 \\
&q_{i,6}^{T} q_{i,6}-v_{i}^{T}(k+1)E_{i,4}^{T} E_{i,4} v_{i}(k+1) \leq 0
\end{aligned}
\end{equation}
which can be written in $\eta_{i,2}(k+1)$ as
\begin{equation}
    \begin{aligned}
&\eta_{i,2}^{\mathrm{T}}(k+1) \operatorname{diag}\{-1, I,0,0,0,0\} \eta_{i,2}(k+1) \leq 0 \\
&\eta_{i,2}^{\mathrm{T}}(k+1)\operatorname{diag}\{-1, 0,R_{i}^{-1}(k+1),0,0,0\}\\
&\quad\quad\quad\quad\quad\quad\quad\quad\quad\times\eta_{i,2}(k+1) \leq 1 \\
&\eta_{i,2}^{\mathrm{T}}(k+1)\operatorname{diag}\{-\hat{x}_{i}^{T}(k+1 \mid k) E_{i,3}^{T} E_{i,3} \hat{x}_{i}(k+1 \mid k),\\
&\quad\quad\quad\quad\quad\quad\quad\quad\quad\quad0,0,I,0,0\}\eta_{i,2}(k+1) \leq 1 \\
&\eta_{i,2}^{\mathrm{T}}(k+1)\operatorname{diag}\{0,-\Xi_{i}^{T}(k+1 \mid k) E_{i,3}^{T} E_{i,3}\\
&\quad\quad\quad\quad\quad\quad\quad\quad\quad\times\Xi_{i}(k+1 \mid k),0,0,I,0\}\\
&\quad\quad\quad\quad\quad\quad\quad\quad\quad\times\eta_{i,2}(k+1) \leq 1\\
&\eta_{i,2}^{\mathrm{T}}(k+1)\operatorname{diag}\{0,0,-E_{i,4}^{T} E_{i,4},0,0,,I\}\\
&\quad\quad\quad\quad\quad\quad\quad\quad\quad\times\eta_{i,2}(k+1) \leq 1
\end{aligned}\label{eq51}
\end{equation}
By applying S-procedure to \eqref{eq44} and \eqref{eq51}, we can conclude that the inequality \eqref{eq44} holds if there exist nonnegative scalars $\tau_{i,11}(k), \tau_{i,12}(k), \tau_{i,13}(k), \tau_{i,14}(k),$ and $\tau_{i,15}(k)$ such that
\begin{equation}
    \begin{aligned}
&\Gamma_{i,3,l_{i}j_{i}}^{T}P_{i}^{-1}(k+1 \mid k+1) \Gamma_{i,3,l_{i}j_{i}}\\&-\operatorname{diag}\{1,0,0,0,0,0\}-\tau_{i, 11}(k) \operatorname{diag}\{-1, I,0,0,0,0\}\\
&-\tau_{i, 12}(k) \operatorname{diag}\{-1, 0,R_{i}^{-1}(k+1),0,0,0\}\\
&-\tau_{i, 13}(k)\operatorname{diag}\{-\hat{x}_{i}^{T}(k+1 \mid k) E_{i,3}^{T} E_{i,3} \hat{x}_{i}(k+1 \mid k),\\
&\quad\quad\quad\quad\quad\quad\quad\quad\quad0,0,I,0,0\}\\
&-\tau_{i, 14}(k)\operatorname{diag}\{0,-\Xi_{i}^{T}(k+1 \mid k) E_{i,3}^{T} E_{i,3}\\
&\quad\quad\quad\quad\quad\quad\quad\quad\quad\times\Xi_{i}(k+1 \mid k),0,0,I,0\}\\
&-\tau_{i, 15}(k) \operatorname{diag}\{0,0,-E_{i,4}^{T} E_{i,4},0,0,,I\}\leq 0.
\end{aligned}\label{eq47}
\end{equation}
Inequality \eqref{eq47} can be written in the following compact form
\begin{equation}
     \begin{aligned}
\Gamma_{i,3,l_{i}j_{i}}^{T}P&_{i}^{-1}(k+1 \mid k+1)\Gamma_{i,3,l_{i}j_{i}}-\operatorname{diag}\{1\\&-\tau_{i,11}(k)-\tau_{i, 12}(k)-\tau_{i, 13}(k)\\&\times \hat{x}_{i}^{T}(k+1 \mid k)E_{i,3}^{T} E_{i,3}\hat{x}_{i}(k+1 \mid k),\\&\tau_{i, 11}(k)I-\tau_{i, 14}(k)\Xi_{i}^{T}(k+1 \mid k) E_{i,3}^{T} E_{i,3}\\&\times \Xi_{i}(k+1 \mid k),\tau_{i, 12}(k)R_{i}^{-1}(k+1)\\&-\tau_{i, 15}(k)E{i,4}^{T}E_{i,4}, \tau_{i, 13}(k)I,\tau_{i, 14}(k)I,\\&\tau_{i, 15}(k)I\}\leq 0.
\end{aligned}\label{eq48}
\end{equation}
Finally denoting
\begin{equation}
\begin{aligned}
\Theta_{i,3}(k)\quad=&\operatorname{diag}\{1-\tau_{i,11}(k)-\tau_{i, 12}(k)-\tau_{i, 13}(k)\\&\times\hat{x}_{i}^{T}(k+1 \mid k) E_{i,3}^{T} E_{i,3}\hat{x}_{i}(k+1 \mid k),\\&\tau_{i, 11}(k)I-\tau_{i, 14}(k)\Xi_{i}^{T}(k+1 \mid k) E_{i,3}^{T} E_{i,3}\\&\times\Xi_{i}(k+1 \mid k),\tau_{i, 12}(k)R_{i}^{-1}(k+1)\\&-\tau_{i, 15}(k)E{i,4}^{T}E_{i,4},\tau_{i, 13}(k)I,\tau_{i, 14}(k)I,\\&\tau_{i, 15}(k)I\},
\end{aligned}\label{eq49}
\end{equation}
we can write \eqref{eq48} as
\begin{equation}
    \begin{aligned}
&\Gamma_{i,3,l_{i}j_{i}}^{T}P_{i}^{-1}(k+1 \mid k+1) \Gamma_{i,3,l_{i}j_{i}}-\Theta_{i,3}(k)\leq 0.
\end{aligned}\label{eq50}
\end{equation}
Now, we deal with the output constraint \eqref{eq21} in Section \ref{sub3-b}. First, it can be described by
\begin{equation}
    \Gamma_{l_{i}y_{i}}\left(\hat{x}_{i}(k+1 \mid k)\right) \eta_{i,2}(k+1)=0.\label{eq56}
\end{equation}
By denoting
\begin{equation}
\begin{aligned}
    &P_{l_{i}y_{i},1}=C_{l_{i}}\hat{x}_{i}(k+1 \mid k)-y_{i}(k+1)\\
    &P_{l_{i}y_{i},2}=C_{l_{i}}\Xi_{i}(k+1 \mid k),
    \end{aligned}
\end{equation}
we have
\begin{equation}
\begin{aligned}
    \Gamma_{l_{i}y_{i}}(\hat{x}_{i}(k+1 \mid k))=\left[P_{l_{i}y_{i},1}\quad P_{l_{i}y_{i},2}\quad D_{l_{i}}\quad H_{i,3} \quad H_{i,3} \quad H_{i,4}\right].
\end{aligned}
\end{equation}
By virtue of Finsler's lemma \cite{b29}, the inequality \eqref{eq44} under constraint \eqref{eq56} holds if there exists a $Z_i(k+1)$ such that
\begin{equation}
    \begin{aligned}
\Gamma_{i,3,l_{i}j_{i}}^{T}&P_{i}^{-1}(k+1 \mid k+1) \Gamma_{i,3,l_{i}j_{i}}-\Theta_{i,3}(k) \\
&+Z_{i}^{\mathrm{T}}(k+1) \Gamma_{l_{i}y_{i}}\left(\hat{x}_{i}(k+1 \mid k)\right)\\
&+\Gamma_{l_{i}y_{i}}^{\mathrm{T}}\left(\hat{x}_{i}(k+1 \mid k)\right) Z_{i}(k+1) \leq 0 .
\end{aligned}\label{eq53}
\end{equation}
For the purpose of simplicity, denote
\begin{equation}
\begin{aligned}
    \Theta_{i,4}(k)=&\Theta_{i,3}(k)-Z_{i}^{\mathrm{T}}(k+1) \Gamma_{l_{i}y_{i}}\left(\hat{x}_{i}(k+1 \mid k)\right)\\
    &-\Gamma_{l_{i}y_{i}}^{\mathrm{T}}\left(\hat{x}_{i}(k+1 \mid k)\right) Z_{i}(k+1).
    \end{aligned}
\end{equation}
Then, by using Schur complements, \eqref{eq53} is equivalent to
\begin{equation}
    \left[\begin{array}{cc}
-P_{i}(k+1 \mid k+1) & \Gamma_{i,3,l_{i}j_{i}} \\
\Gamma_{i,3,l_{i}j_{i}}^{\mathrm{T}} & -\Theta_{i,4}(k)
\end{array}\right] \leq 0.\label{eq55}
\end{equation}
\begin{theorem}
Consider the leader following multi-agent system \eqref{eq1}, \eqref{eq11} that satisfies Assumption \ref{assump1}, Assumption \ref{assump2} and Assumption \ref{assump3}. If the state $x_{i}(k+1)$ belongs to its state prediction ellipsoid $\left(x_{i}(k+1)-\hat{x}_{i}(k+1 \mid k)\right)^{\mathrm{T}} P_{i}^{-1}(k+1 \mid k)\left(x_{i}(k+1)-\hat{x}_{i}(k+1 \mid k)\right) \leq 1$, then such a state also resides in its updated state estimation ellipsoid $\left(x_{i}(k+1)-\hat{x}_{i}^{\mathrm{T}}(k+1 \mid k+1)\right) P_{i}^{-1}(k+1 \mid k+1)\left(x_{i}(k+1)-\hat{x}_{i}(k+1 \mid k+1)\right) \leq 1$ with the center determined by \eqref{eq17}, where $P_{i}(k+1 \mid k+1)>0$ satisfies matrix inequality \eqref{eq55} with other decision variables $L_{l_{i}}(k+1)$, $N_{i}(k+1),$ and $\tau_{i,m}(k) \geq 0$ for $m=11,\ldots,15$ for all $l_{i},j_{i}=1,\dots,r$.
\begin{proof}
According to the above discussion, if there exist $P_{i}(k+1 \mid k+1)>0, L_{l_{i}}(k+1)$, $N_{i}(k+1),$ and $\tau_{i,m}(k) \geq 0$ for $m=11,\ldots,15$ such that \eqref{eq55} holds for all $l_{i},j_{i}=1,\dots,r$, then we have
\begin{equation}
\begin{aligned}
    \sum_{l_{i}=1}^{r}&g_{l_{i}}\left(\theta_{i}(k)\right)\sum_{j_{i}=1}^{r} g_{j_{i}}\left(\hat{\theta}_{i}(k)\right)\\&\eta_{i,2}^{\mathrm{T}}(k+1)\Gamma_{i,3,l_{i}j_{i}}^{\mathrm{T}}P_{i}^{-1}(k+1\mid k+1)\Gamma_{i,3,l_{i}j_{i}}\eta_{i,2}(k+1) \leq 1.
    \end{aligned}
\end{equation}
From \eqref{eq104}, we obtain
\begin{equation}
\begin{aligned}
    \left(x_{i}(k+1)\right.&\left.-\hat{x}_{i}(k+1\mid k+1)\right)^{\mathrm{T}} P_{i}^{-1}(k+1 \mid k+1)\\
        &\times\left(x_{i}(k+1)-\hat{x}_{i}(k+1 \mid k+1)\right) \leq 1,
    \end{aligned}
\end{equation}
which completes the proof.
\end{proof}

Now, the convex optimization approach is applied to determine an optimal ellipsoid with the minimal size. Therefore, $P_{i}(k+1 \mid k+1)$ is obtained by solving the following optimization problem:
\begin{equation}
\begin{aligned}
    \underset{\begin{subarray}{c}
    P_{i}(k+1 \mid k+1), L_{l{i}}(k+1),\\
    \tau_{i,11}(k), \tau_{i,12}(k), \tau_{i,13}(k),\\
  \tau_{i,14}(k), \tau_{i,15}(k), Z_{i}(k+1)\\
  \end{subarray}}{\operatorname{min}}\operatorname{Tr}\left(P_{i}(k+1 \mid k+1)\right)\\
\end{aligned}\label{eq58}
\end{equation}
subject to \eqref{eq55}.
\label{theo2}
\end{theorem}
\subsection{Recursive Algorithm for Attack Diagnosis}
The recursive algorithm based on the set-membership filtering to compute the state ellipsoids so that a cyberattack can be detected is summarized below.

Algorithm 1 recursively computes the prediction ellipsoid $\mathcal{X}_{i}(k+1 \mid k)$ and its update $\mathcal{X}_{i}(k+1 \mid k+1)$ with the current measurement $y_{i}(k+1)$. Steps 3 and 6 of the algorithm are proposed to detect cyberattacks that affect control signals, communication channels and sensor measurements.
\begin{table}[t]
\begin{tabular}{p{0.94\columnwidth}l}
\toprule\textbf{Algorithm 1 Recursive State Estimation} \\
\midrule \textbf{1. Initialization:} \\
\midrule Given an initial ellipsoid $\mathcal{X}_{i}(0 \mid 0), \mathcal{U}_{i}(0)$, recursive times $T_{N}$, and set $k=0$. Let $x_{i}(k)=x_{i}(0), \hat{x}_{i}(k \mid k)=\hat{x}_{i}(0 \mid 0), \Xi_{i}(k \mid k)=\Xi_{i}(0 \mid 0), x^{l}(k)=x^{l}(0),$ and $\xi_{i}(k)=\xi_{i}(0)$. \\
\midrule \textbf{2. Prediction:} \\
\midrule 1) Calculate $P_{i}(k+1 \mid k), U_{i}(k+1), \hat{A}_{l_{i}}(k), K_{l_{i}}$ by solving the optimization problem \eqref{eq39}. \\
2) Obtain the matrix $\Xi_{i}(k+1 \mid k),$ and $\xi_{i}(k+1 \mid k)$ according to $P_{i}(k+1 \mid k)=\Xi_{i}(k+1 \mid k)\Xi_{i}^\mathrm{T}(k+1 \mid k),$ and $U_{i}(k+1)=\xi_{i}(k+1)\xi_{i}^\mathrm{T}(k+1)$. \\
3) Calculate the centre of the prediction ellipsoid $\hat{x}_{i}(k+1 \mid k)$ by \eqref{eq13}. \\
\midrule \textbf{3. Attack Detection:} Control Signal Data or Communication Signal Data Cyber Attack Diagnosis \\
\midrule 1) If $\mathcal{X}_{i}(k \mid k) \bigcap \mathcal{X}_{i}(k+1 \mid k) \neq \varnothing$, there is no attack and go to step 5. \\
2) If $\mathcal{X}_{i}(k \mid k) \cap \mathcal{X}_{i}(k+1 \mid k)=\varnothing$, data is subject to attack and go to step 4. \\
\midrule \textbf{4. Recovery Step and Attack Mitigation:}\\
\midrule Set $\mathcal{X}_{i}(k+1 \mid k) \leftarrow \mathcal{X}_{i}(k \mid k),$ $\mathcal{U}_{i}(k+1) \leftarrow \mathcal{U}_{i}(k), u_{i}(k) \leftarrow u_{i}(k-1),$
and go to step 5.\\
\midrule \textbf{5. Measurement Update:} \\
\midrule 1) Calculate $P_{i}(k+1 \mid k+1)$ and $L_{l_{i}}(k+1)$ by solving the optimization problem \eqref{eq58}.\\
2) Obtain the new $\Xi_{i}(k+1 \mid k+1)$ according to $P_{i}(k+1 \mid k+1)=\Xi_{i}(k+1 \mid k+1)\Xi_{i}^\mathrm{T}(k+1 \mid k+1)$.\\
3) Calculate the centre of the updated estimation ellipsoid $\hat{x}_{i}(k+1 \mid k+1)$ by \eqref{eq17}.\\
\midrule \textbf{6. Attack Detection:} Sensor Measurement Data Cyber Attack Diagnosis\\
\midrule 1) If $\mathcal{X}_{i}(k+1 \mid k+1) \bigcap {\mathcal{X}_{i}(k+1 \mid k)} \neq \varnothing$, there is no attack and go to step 8. \\
2) If $\mathcal{X}_{i}(k+1 \mid k+1) \bigcap\mathcal{X}_{i}(k+1 \mid k)=\varnothing$, data is subject to attack and go to step 7.\\
\midrule \textbf{7. Recovery Step and Attack Mitigation:}\\
\midrule Set $\mathcal{X}_{i}(k+1 \mid k+1) \leftarrow \mathcal{X}_{i}(k+1 \mid k), y_i(k+1) \leftarrow y_i(k)$ and go to step 8. \\
\midrule \textbf{8. Loop} \\
\midrule If $k==T_{N}$ then Exit, Else $k \leftarrow k+1$ and go to step 2.\\
\bottomrule
\end{tabular}
\end{table}
\section{Simulation Results}
Consider the following multi-agent, discrete-time, nonlinear system:
\begin{equation}
    \begin{aligned}
        x_{1,1}(k+1) =&0.2x_{1,1}(k)-0.3\left(x_{1,2}(k)-x_{1,1}^{2}(k)\right)\\
        &+u_{1}(k)+w_{k} \\
        x_{1,2}(k+1) =&0.3 x_{1,1}(k)+0.2\left(x_{1,2}(k)-x_{1,1}^{2}(k)\right)\\
        &+0.3u_{1,1}(k)+0.9u_{1,2}(k)+\omega_{k} \\
        y_{k} =&x_{1,1}(k)+0.1 x_{1,1}^{2}(k)+x_{1,2}(k)+v_{k}\\
        x_{2,1}(k+1) =&0.5x_{2,1}(k)-0.1\left(x_{2,2}(k)-x_{2,1}^{2}(k)\right)\\
        &+0.9u_{2,1}(k)+0.2u_{2,2}(k)+w_{k} \\
        x_{2,2}(k+1) =&0.9x_{2,1}(k)+0.5\left(x_{2,2}(k)-x_{2,1}^{2}(k)\right)\\
        &+u_{2}(k)+\omega_{k} \\
        y_{k} =&x_{2,1}(k)+0.1 x_{2,1}^{2}(k)+x_{2,2}(k)+v_{k}
    \end{aligned}
\end{equation}
where the state $x_{i}(k)=\left[\begin{array}{ll}x_{i,1}(k) & x_{i,2}(k)\end{array}\right]^{T}$.
Now, we construct the following fuzzy models to approximate the above nonlinear multi-agent system for each agent:\\
\textbf{Agent 1:}\\
- Rule 1: IF $x_{1, 1}(k)$ is about 1,THEN
\begin{equation}
\begin{aligned}
    x_{1}(k+1)=A_{1,1} x_{1}(k)+B_{1,1} u_{1}(k)+M_{1,1}\omega_{1}(k),\\
    y_{1}(k)=C_{1,1}x_{1}(k)+D_{1,1}v_{1}(k)
    \end{aligned}
\end{equation}
- Rule 2: IF $x_{1,1}(k)$ is about 0, THEN
\begin{equation}
\begin{aligned}
    x_{1}(k+1)=A_{1,2} x_{1}(k)+B_{1,2} \omega_{1}(k)\\
    y_{1}(k)=C_{1,2} x_{1}(k)+D_{1,2} v_{1}(k)
    \end{aligned}
\end{equation}
where
$$
\begin{array}{rlr}
A_{1,1} & =\left[\begin{array}{cc}
0.5 & -0.3 \\
0.1 & 0.2
\end{array}\right] \quad B_{1,1}=\left[\begin{array}{l}
1 \\
1
\end{array}\right]\\ C_{1,1}&=\left[\begin{array}{lll}
1.1 & 1.1
\end{array}\right] \quad D_{1,1}=1 \\
A_{1,2}&=\left[\begin{array}{cc}
0.2 & -0.3 \\
0.3 & 0.2
\end{array}\right] \quad B_{1,2}=\left[\begin{array}{l}
1 \\
1
\end{array}\right] \\ C_{1,2}&=\left[\begin{array}{lll}
1.0 & 1.0
\end{array}\right] \quad D_{1,2}=1.
\end{array}
$$
\textbf{Agent 2:}\\
- Rule 1: IF $x_{2, 1}(k)$ is about 1,THEN
\begin{equation}
\begin{aligned}
    x_{2}(k+1)=A_{2,1} x_{2}(k)+B_{2,1} u_{2}(k)+M_{2,1}\omega_{1}(k),\\
    y_{2}(k)=C_{2,1}x_{2}(k)+D_{2,1}v_{2}(k)
    \end{aligned}
\end{equation}
- Rule 2: IF $x_{2,1}(k)$ is about 0, THEN
\begin{equation}
\begin{aligned}
    x_{2}(k+1)=A_{2,2} x_{2}(k)+B_{2,2} \omega_{2}(k)\\
    y_{2}(k)=C_{2,2} x_{2}(k)+D_{2,2} v_{2}(k)
    \end{aligned}
\end{equation}
where
\begin{equation}
\begin{array}{rlr}
A_{2,1} & =\left[\begin{array}{cc}
0.6 & -0.1 \\
0.4 & 0.5
\end{array}\right] \quad B_{2,1}=\left[\begin{array}{l}
1 \\
1
\end{array}\right]\\ C_{2,1}&=\left[\begin{array}{lll}
1.1 & 1.1
\end{array}\right] \quad D_{2,1}=1 \\
A_{2,2}&=\left[\begin{array}{cc}
0.5 & -0.1 \\
0.9 & 0.5
\end{array}\right] \quad B_{2,2}=\left[\begin{array}{l}
1 \\
1
\end{array}\right] \\ C_{2,2}&=\left[\begin{array}{lll}
1.0 & 1.0
\end{array}\right] \quad D_{2,2}=1.
\end{array}
\end{equation}
For the convenience of simulation, triangular membership functions are used for Rule 1 and Rule 2 in this example.

In the above fuzzy models, the approximation errors between the nonlinear system and the fuzzy models are assumed to satisfy \eqref{eq12}, where
\begin{equation}
\begin{aligned}
\begin{array}{ll}
H_{1,1}=\left[\begin{array}{c}
0.1 \\
0.1
\end{array}\right] & H_{2,1}=\left[\begin{array}{c}
0.3 \\
0.3
\end{array}\right] \\ E_{1,1}=\left[\begin{array}{ll}
0 & 0.5
\end{array}\right] & E_{2,1}=\left[\begin{array}{ll}
0 & 0.6
\end{array}\right] \\ H_{i,2}=\left[\begin{array}{l}
0 \\
0
\end{array}\right] & E_{i,2}=0 \\
H_{i,3}=0.1 & E_{i,3}=\left[\begin{array}{ll}
0 & 0.5
\end{array}\right] \\ H_{i,4}=0 & E_{i,4}=0
\end{array}
\end{aligned}
\end{equation}
The leader matrices described in \eqref{ext} is defined as follow:
\begin{equation}
\begin{array}{ll}
A_{1}^{l}  =\left[\begin{array}{cc}
0.5 & 0.2 \\
-0.6 & 0.7
\end{array}\right]&
A_{2}^{l}  =\left[\begin{array}{cc}
0.5 & 0.2 \\
-0.4 & 0.7
\end{array}\right]
\end{array}
\end{equation}

In the simulation, $\omega_{i}(k)$ and $v_{i}(k)$ are chosen as $0.5 \sin (2 k)$ and $0.5 \sin (20k)$, respectively. The initial state is set as $x_{i}(0)=\left[\begin{array}{ll}0 & 0\end{array}\right]^{T}$, which belongs to the ellipsoids $\left(x_{i}(0\mid 0)-\hat{x}_{i}(0)\right)^{T} P_{i}^{-1}(0\mid 0)\left(x_{i}(0)-\hat{x}_{i}(0\mid 0)\right) \leq 1$ and $\left(x_{i}(0)-{x}^{l}(0)\right)^{T} U_{i}^{-1}(0)\left(x_{i}(0)-{x}_{l}(0)\right) \leq 1$, where $\hat{x}_{i}(0)={x}^{l}(0)=\left[\begin{array}{ll}1 & 1\end{array}\right]^{T}$, and $P_{i}(0\mid 0)=U_{i}(0)=\left[\begin{array}{cc}100 & 0 \\ 0 & 100\end{array}\right], Q_{i}(k)=1-k / 50$, and $R_{i}(k)=1-k / 50$.
The communication between the agents and the leader is modeled as Fig. \ref{fig7}
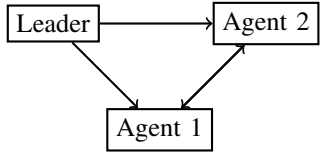
\begin{figure}
\begin{center}
\begin{tikzpicture}[-latex ,auto ,node distance =2 cm and 2 cm, thick , state/.style ={rectangle, draw, minimum width =1 cm}]
\node[state] (1) {Agent 1};
\node[state] (2) [above right of=1] {Agent 2};
\node[state] (3) [above left of=1] {Leader};
\draw[->] (1) -- (2);
\draw[->] (2) -- (1);
\draw[->] (3) -- (1);
\draw[->] (3) -- (2);
\end{tikzpicture}
\caption{Multi-agent system with a leader.}
\label{fig7}
\end{center}
\end{figure}

We obtained the simulation results under MATLAB 9.8 with YALMIP and SDPT3. We considered the following scenarios in 50 sampling steps.
\subsection{Attack Free System}
In this case, the prediction ellipsoid set and the updated estimation ellipsoid set must always have the intersection. Fig. \ref{fig1} (a) and Fig. \ref{fig1} (b)
show the existence of the intersection between these sets for agent 1 and agent 2, respectively.
\begin{figure}[thpb]
  \centering
  \begin{tabular}{ c @{\hspace{20pt}}}
    \includegraphics[width=\columnwidth]{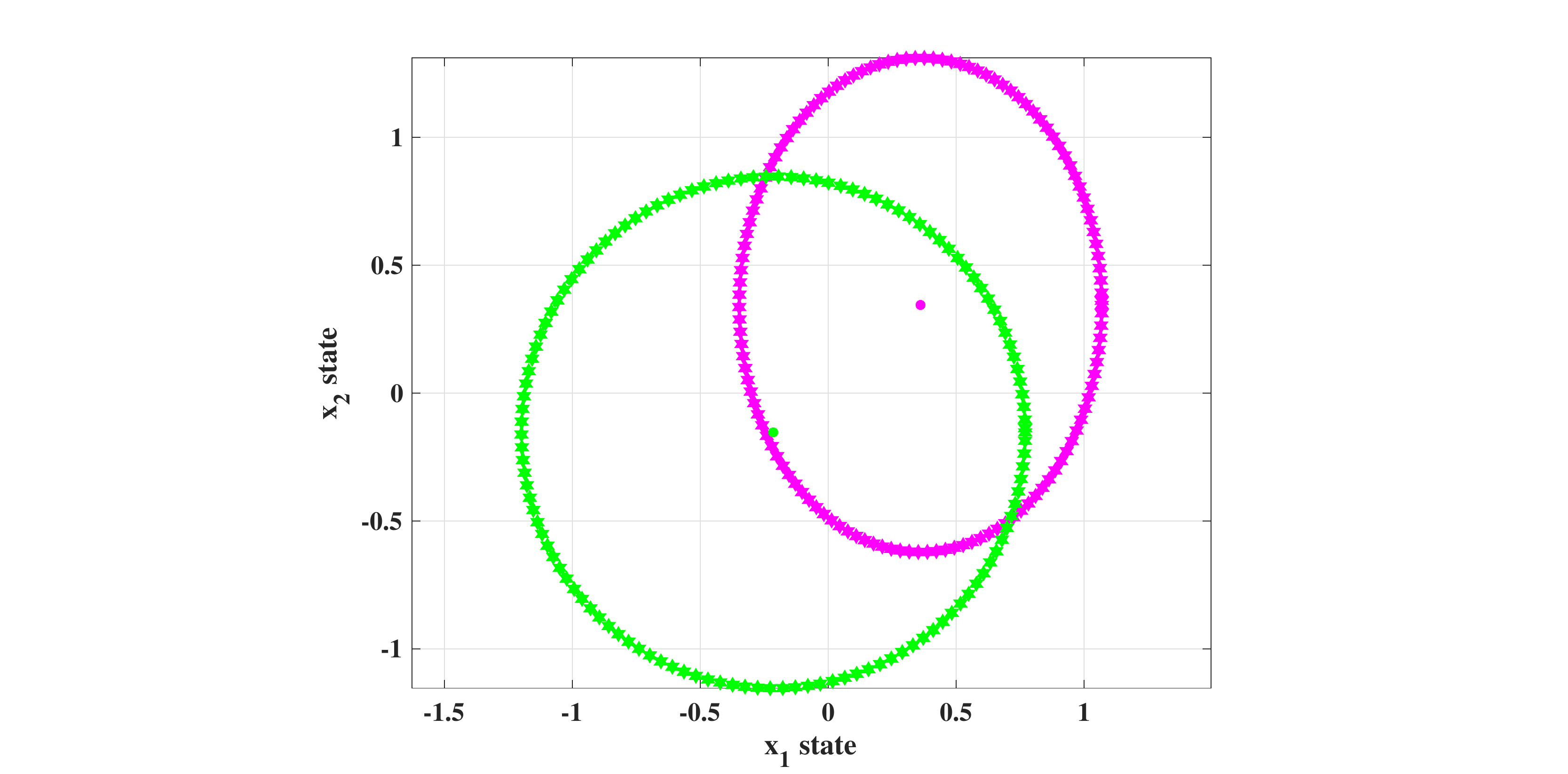} \\
    \small (a) Agent 1\\
      \includegraphics[width=\columnwidth]{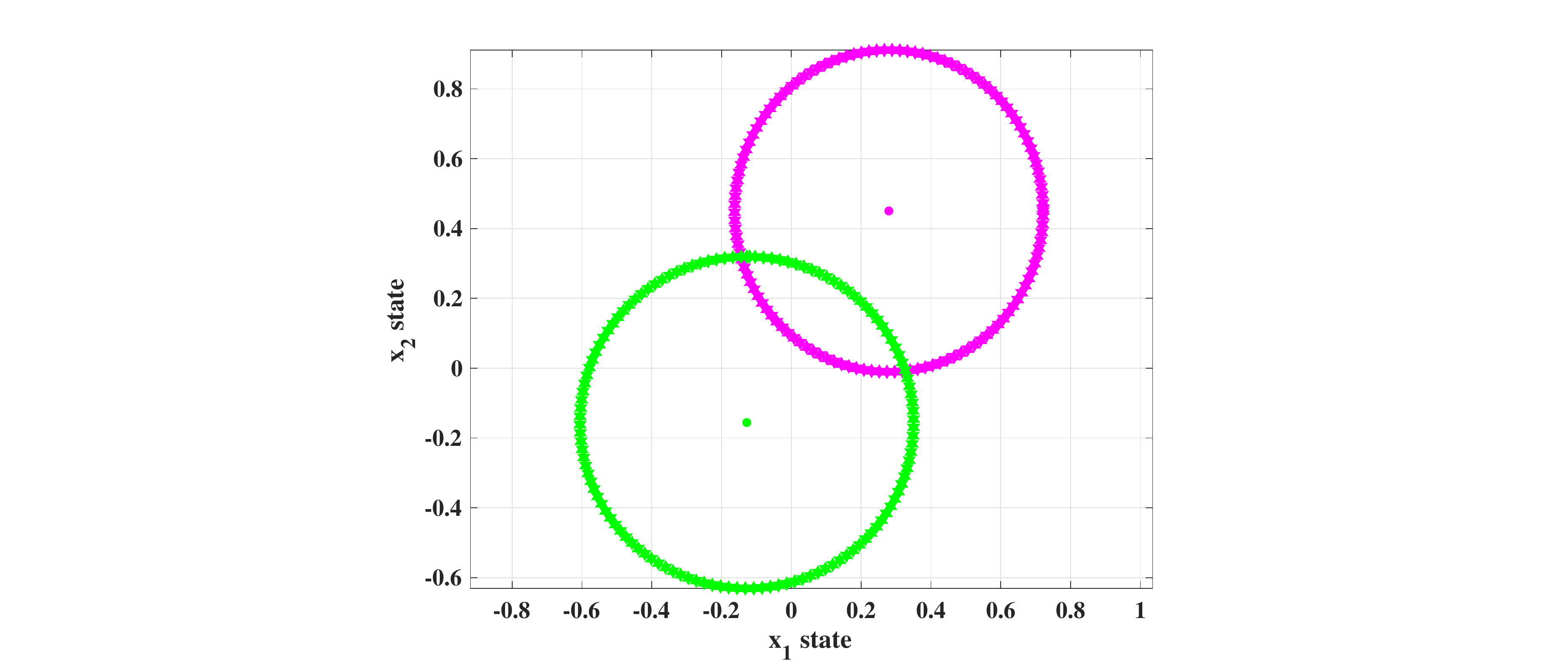} \\
      \small (b) Agent 2
  \end{tabular}
  \medskip

  \caption{Prediction ellipsoid $\mathcal{X}_{i}(k+1 \mid k)$ (pink) and the updated estimation ellipsoid $\mathcal{X}_{i}(k+1 \mid k+1)$ (green) at k=22.}
   \label{fig1}
\end{figure}

\subsection{Replay Attacks on Sensor Data}
We consider adding the replay attack on the sensor measurement data based on the definition in the Subsection \ref{sub2-a}. Therefore, we assume that the attack records the data from k=5 to k=10 and replaces the sensor data at k=20 to k=25 with them. Therefore, Fig. \ref{fig2} (a) and Fig. \ref{fig2} (b) confirm that the prediction ellipsoid set and the updated measurement set for the next iteration do not have the intersection during this attack period.
\begin{figure}[thpb]
  \centering
  \begin{tabular}{ c @{\hspace{20pt}}}
    \includegraphics[width=\columnwidth]{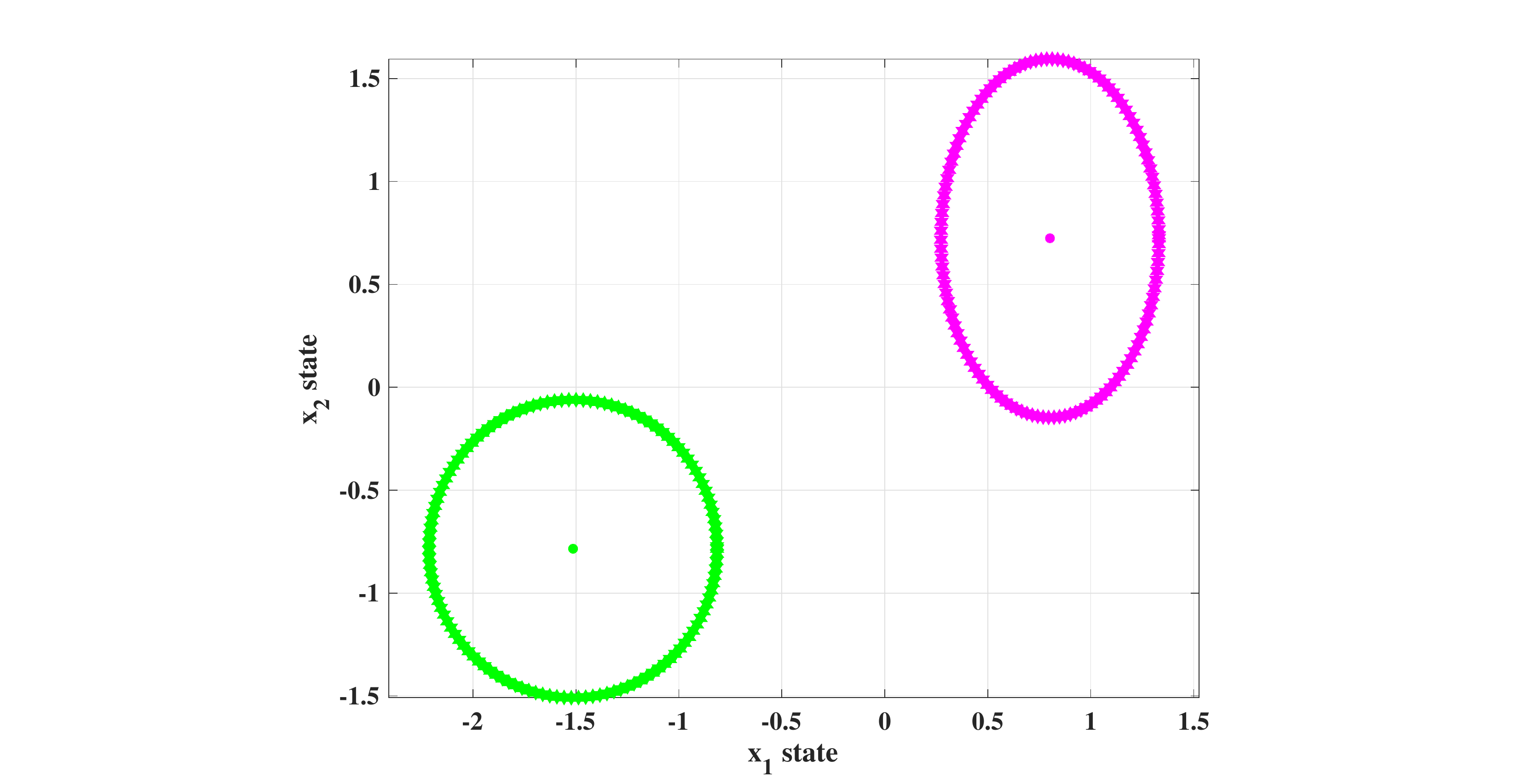} \\
    \small (a) Agent 1\\
      \includegraphics[width=\columnwidth]{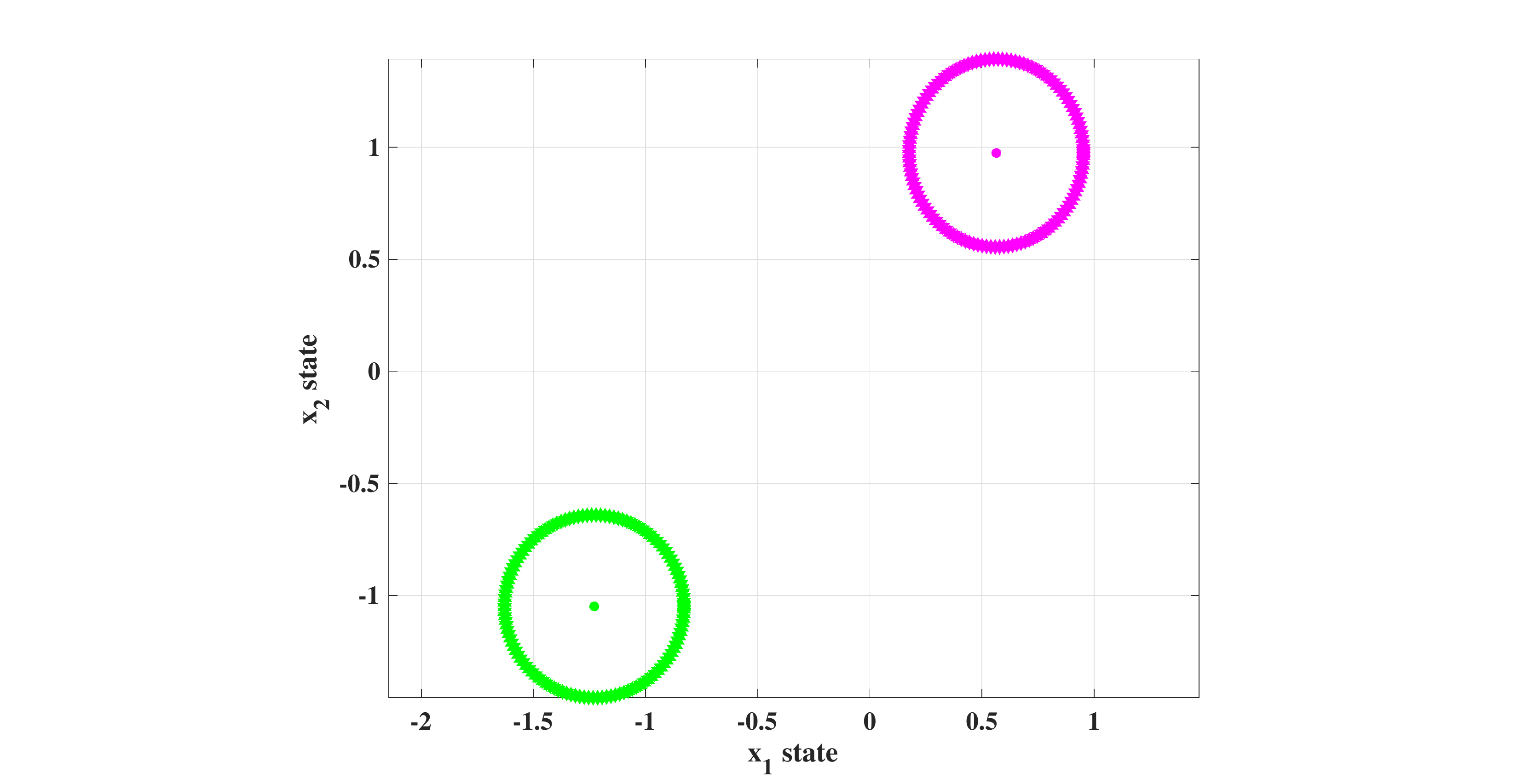} \\
      \small (b) Agent 2
  \end{tabular}
  \medskip

  \caption{Prediction ellipsoid $\mathcal{X}_{i}(k+1 \mid k)$ (pink) and the updated estimation ellipsoid $\mathcal{X}_{i}(k+1 \mid k+1)$ (green) at k=23.}
   \label{fig2}
\end{figure}
\subsection{False Data Injection Attacks on Control Signal}
In this case, the original data packets are replaced by false ones when they are transferred from controllers to actuators via communication channels. Therefore, we consider that the attacker replaces the control signal via targeting the communication channel between the controller and actuator. In the simulation, the attack vector in \eqref{eq3} is modeled as $u_i^a(k)= [4\quad4]^{T}$ for both agents from step k = 20 to k = 25. Therefore, as it is shown in Fig. \ref{fig3} (a) and Fig. \ref{fig3} (b), when the attack on the control signal occurs, the prediction ellipsoid set and the estimation ellipsoid set updated with the previous time instant has no intersection.
\begin{figure}[thpb]
  \centering
  \begin{tabular}{ c @{\hspace{20pt}}}
    \includegraphics[width=\columnwidth]{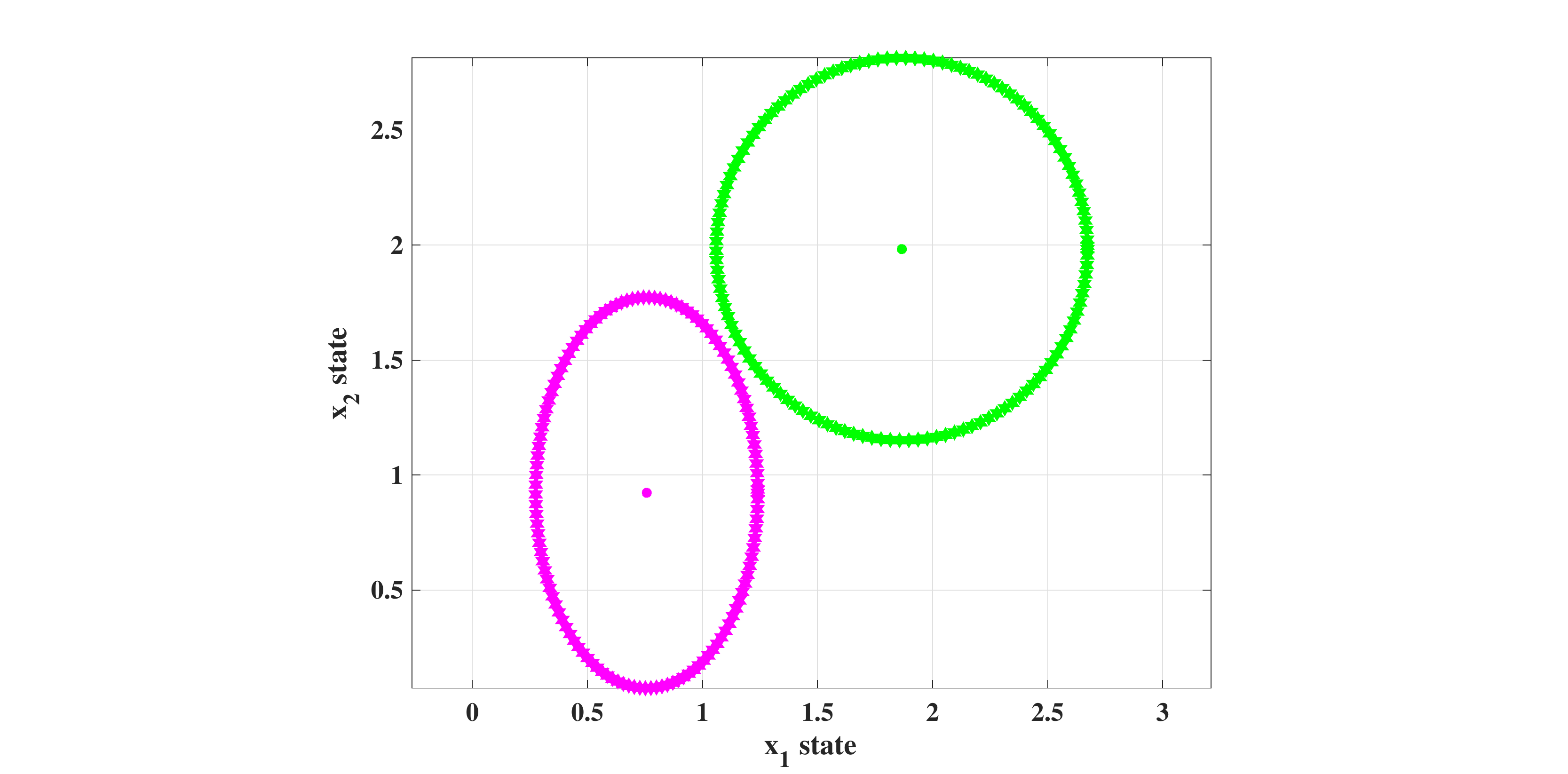} \\
    \small (a) Agent 1\\
      \includegraphics[width=\columnwidth]{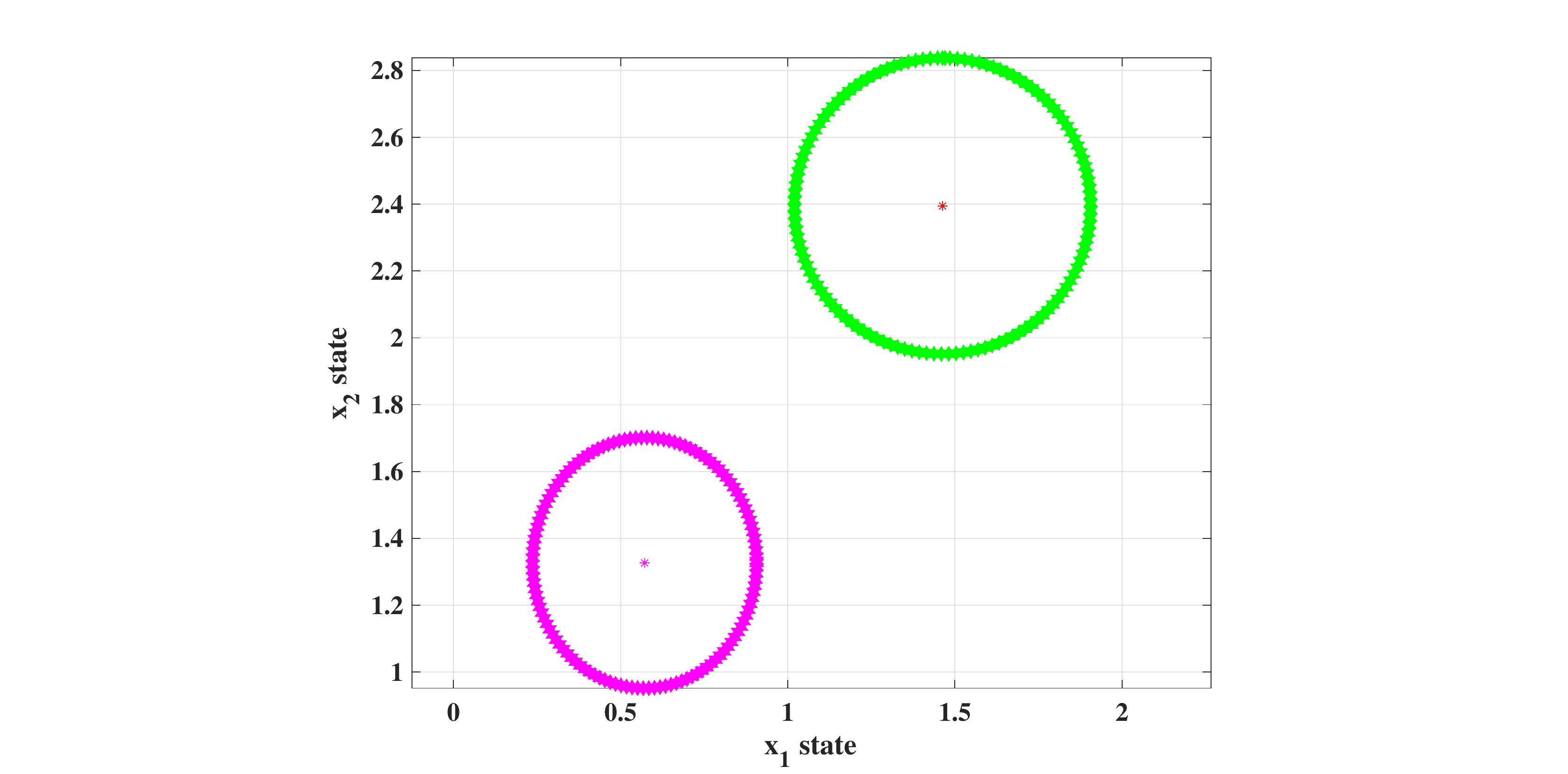} \\
      \small (b) Agent 2
  \end{tabular}
  \medskip

  \caption{Prediction ellipsoid $\mathcal{X}_{i}(k+1 \mid k)$ (pink) and the previous updated estimation ellipsoid $\mathcal{X}_{i}(k \mid k)$ (green) at k=24.}
   \label{fig3}
\end{figure}
\subsection{False Data Injection Attacks on Communication Channel}
In this case, the original data packets were replaced by false ones when they were transferred between two agents via communication channels. Therefore, we consider that the attacker replace the signal transferring from agent 2 to agent 1 via communication channel. In the simulation, the attack vector in \eqref{eq3} is modeled as $\bar{x}_2^c(k)= [5\quad5]^{T}$ from step k=20 to k=25. Therefore, as it is shown in Fig. \ref{fig4}, for agent 1, the prediction ellipsoid set and the estimation ellipsoid set updated with the previous time instant has no intersection. Also, as there is no attack on the communication channel from agent 1 to agent 2, the prediction ellipsoid set and the estimation ellipsoid set updated with the previous time instant for agent 2 has an intersection.
\begin{figure}[thpb]
  \centering
  \begin{tabular}{ c @{\hspace{20pt}}}
    \includegraphics[width=\columnwidth]{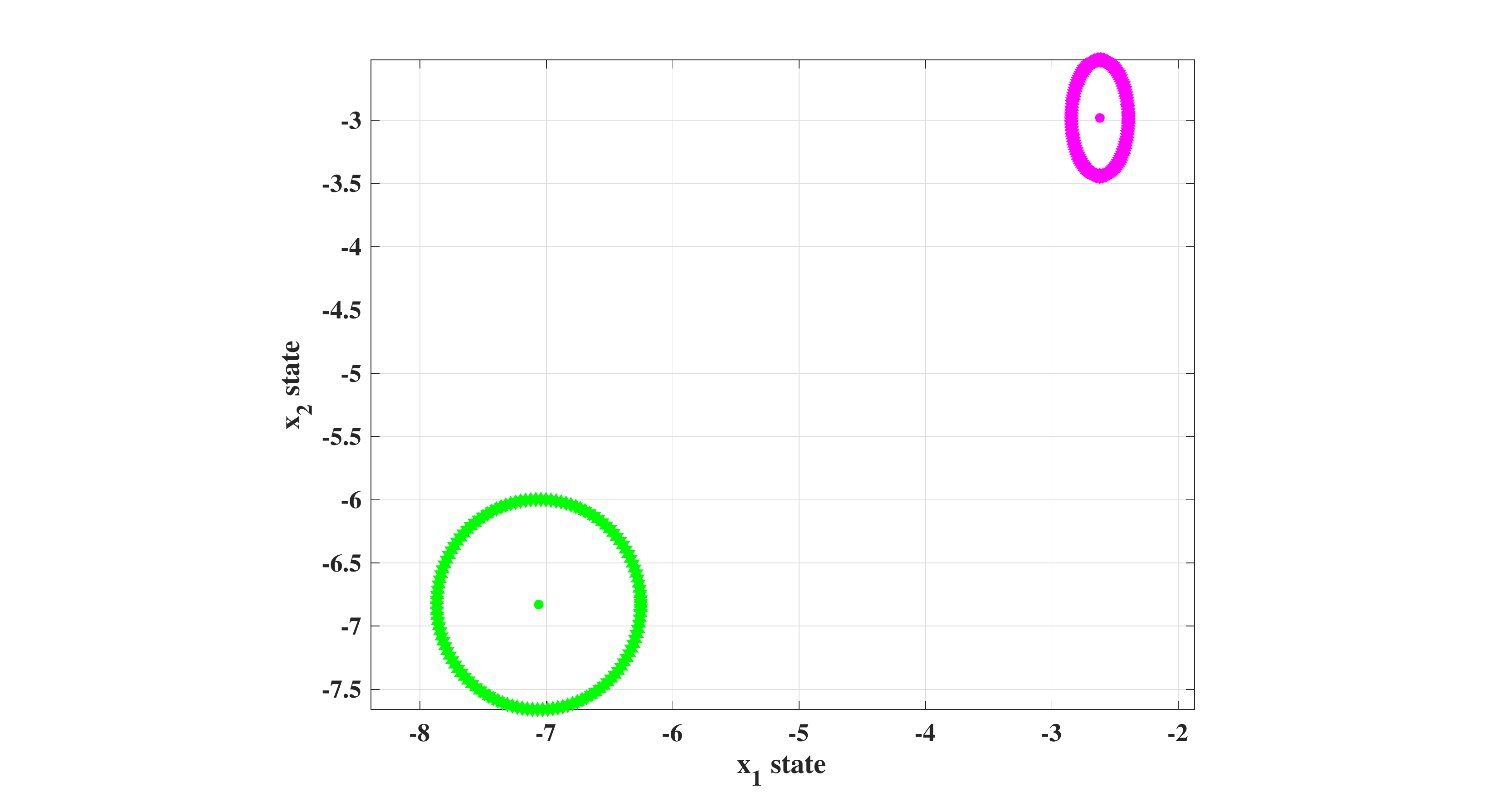} \\
    \small (a) Agent 1\\
      \includegraphics[width=\columnwidth]{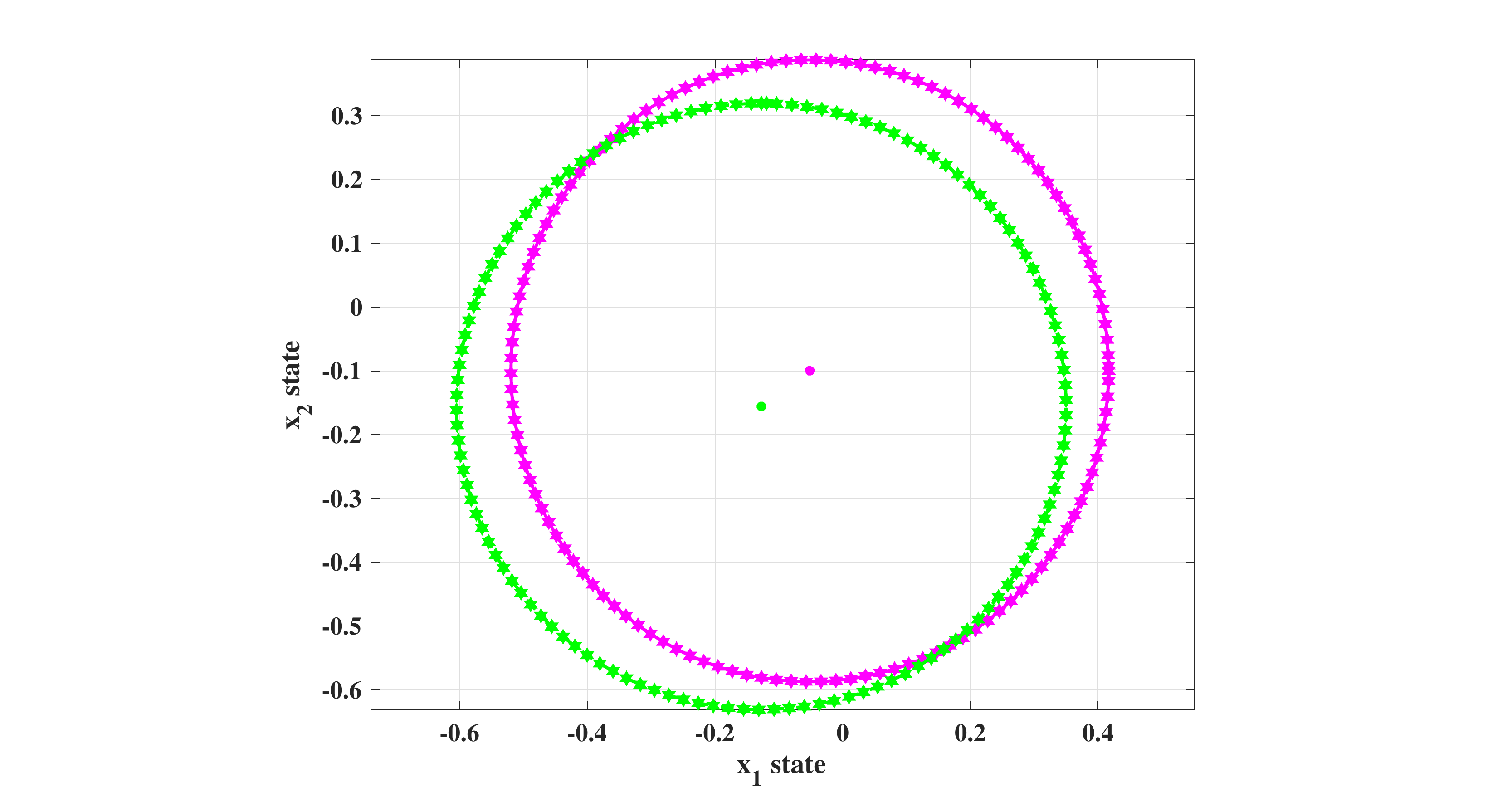} \\
      \small (b) Agent 2
  \end{tabular}
  \medskip

  \caption{Prediction ellipsoid $\mathcal{X}_{i}(k+1 \mid k)$ (pink) and the previous updated estimation ellipsoid $\mathcal{X}_{i}(k \mid k)$ (green) at k=23.}
   \label{fig4}
\end{figure}

Finally, Fig. \ref{fig5} illustrates that we can guarantee the leader following consensus in the attack free system and in the presence of the attacks for all the above scenarios.
\begin{figure}[thpb]
  \centering
  \begin{tabular}{ c @{\hspace{20pt}}}
    \includegraphics[width=\columnwidth]{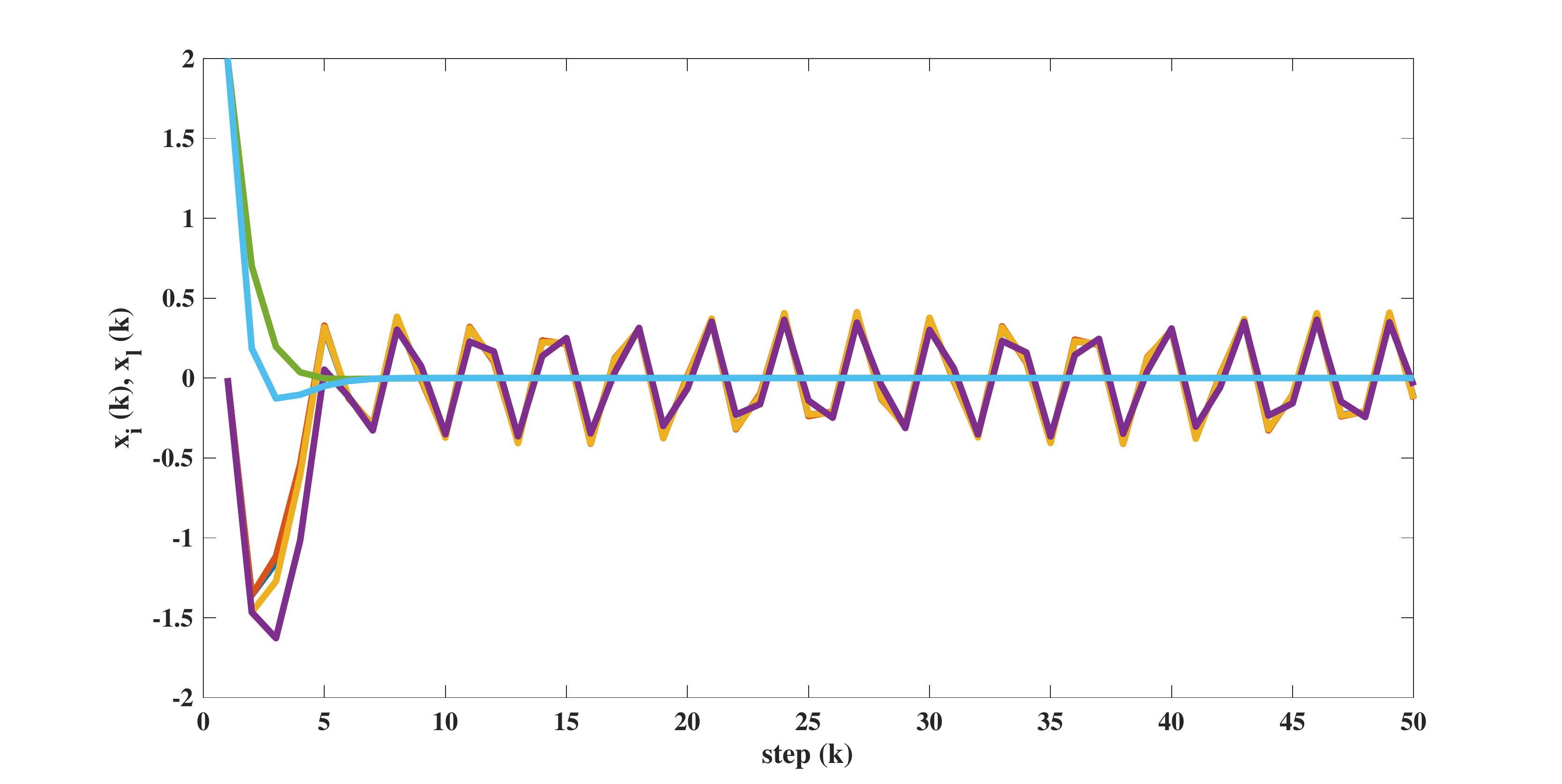} \\
    \small (a) Attack free system\\
      \includegraphics[width=\columnwidth]{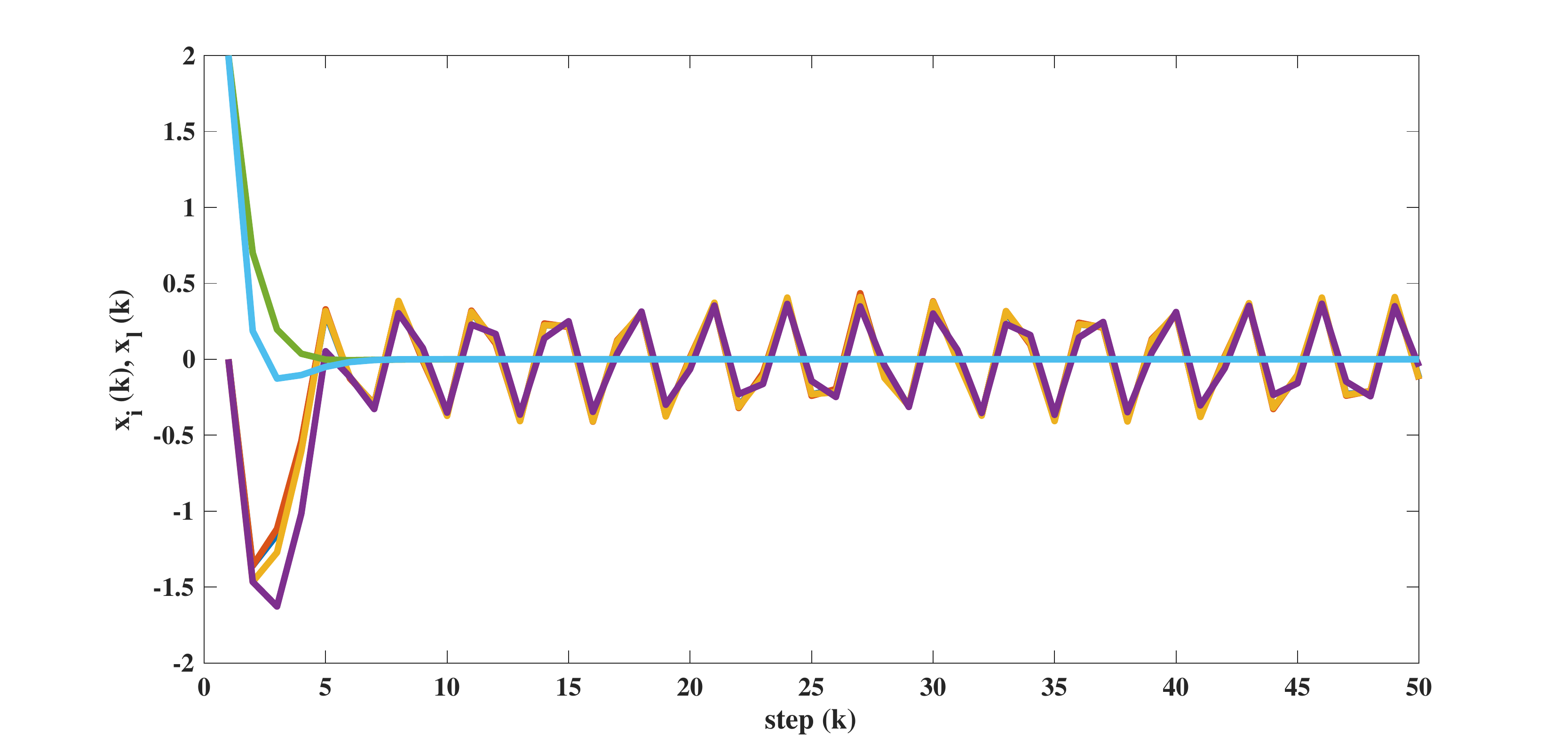} \\
      \small (b) Replay attack on sensor data\\
      \includegraphics[width=\columnwidth]{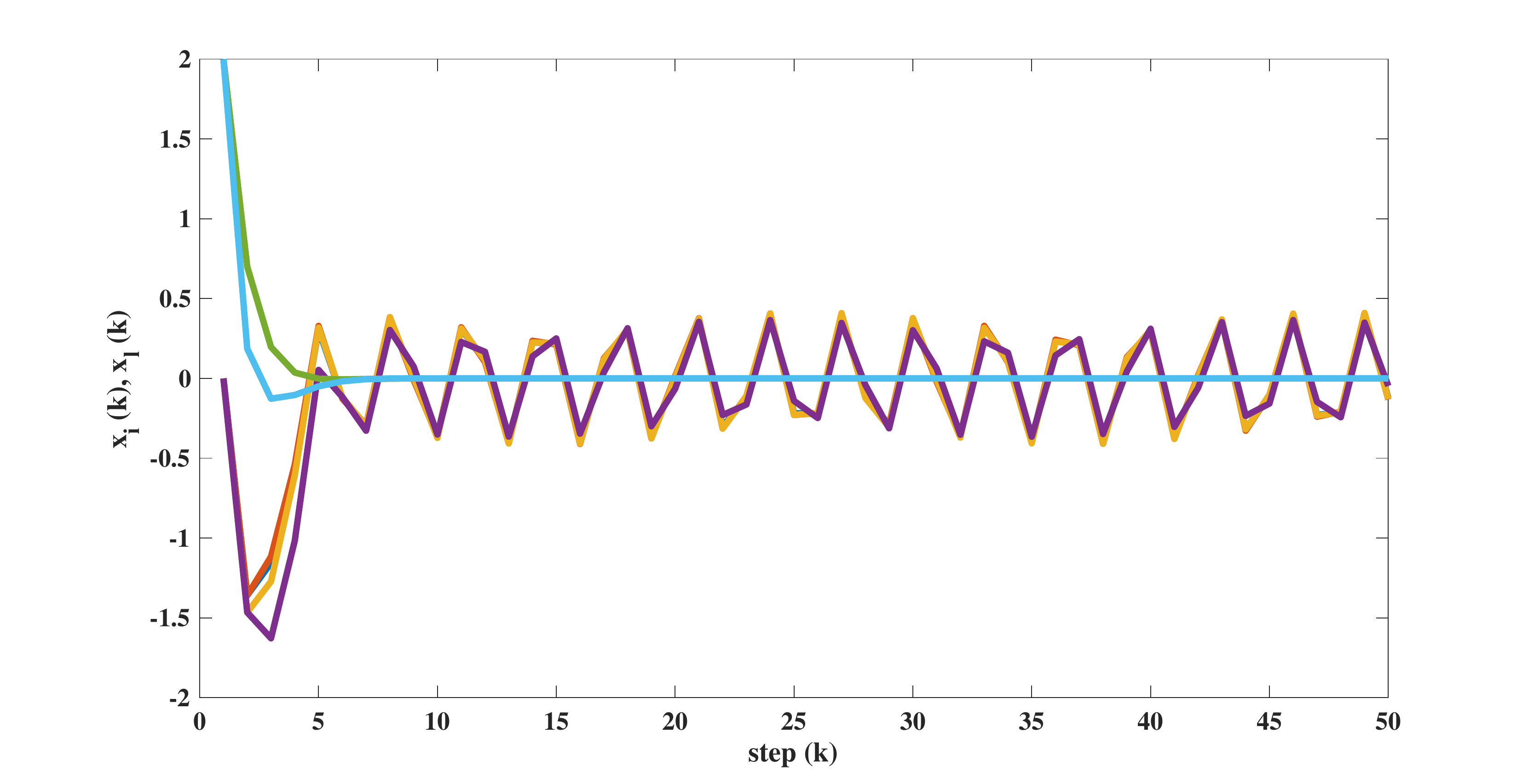} \\
    \small (c) False data injection attack on control signal\\
      \includegraphics[width=\columnwidth]{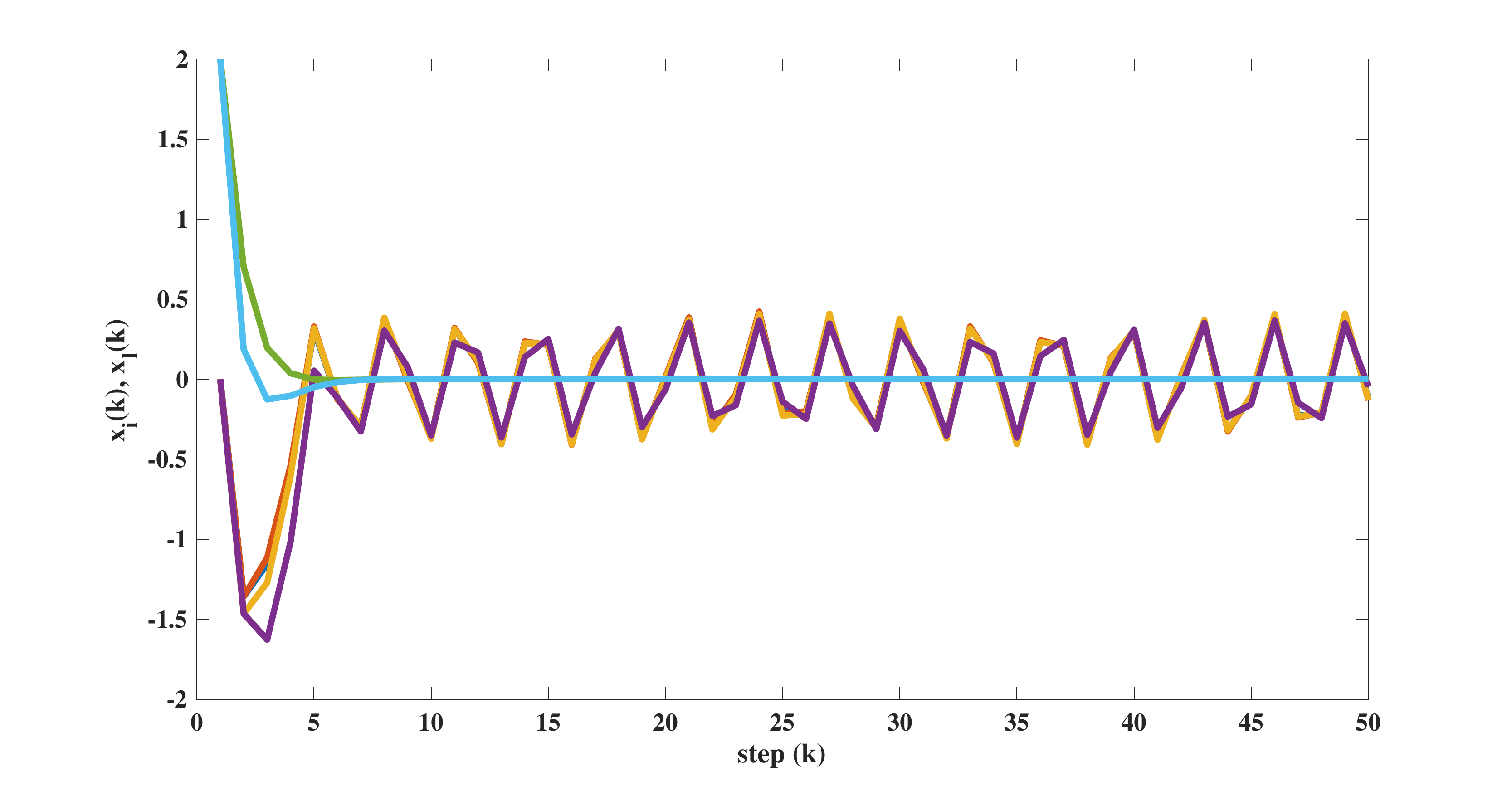} \\
      \small (d) False data injection attack on communication channel
  \end{tabular}
  \medskip
  
  \caption{Leader-following consensus achievements of the agents in all scenarios.}
   \label{fig5}
\end{figure}
\section{Conclusion}
This paper deals with the problem of cyberattack detection in discrete-time leader-following nonlinear multi-agent systems subject to unknown but bounded system noises. For the approximation of the nonlinear systems over the true value of the state, the T-S fuzzy model has been employed. A new fuzzy set-membership filtering method consisting of two steps has been developed for each agent to detect two types of cyberattacks at the time of their occurrence. It has considered the detection of replay attacks and false data injection attacks affecting the leader-following consensus. We proposed recursive algorithms for achieving the consensus protocol and finding the two ellipsoid sets for detecting attacks. The cyberattacks have been detected based on two criteria about the intersections between the ellipsoid sets. Finally, simulation results have been provided to demonstrate the effectiveness of the proposed method.
\bibliographystyle{ieeetr}
\bibliography{bibfile}
\end{document}